\documentclass[a4paper,12pt,openany]{article}
\usepackage{calc}
\usepackage[all]{xy}
\usepackage[centertags]{amsmath}
\usepackage{latexsym}
\usepackage{amsfonts}
\usepackage{cases}
\usepackage{amssymb}
\usepackage{amsthm}
\usepackage{fancyhdr}
\usepackage [dvips]{epsfig}
\fancypagestyle{plain}{ 
\fancyhead{} 
}

\usepackage{newlfont}
\newlength{\defbaselineskip}
\newcommand{\setlinespacing}[1]%
           {\setlength{\baselineskip}{#1 \defbaselineskip}}

\newtheorem{theorem}{Th\'eor\`{e}me}[section]
\newtheorem{lemma}{Lemma}[section]
\newtheorem{proposition}{Proposition}[section]

\theoremstyle{plain}

\usepackage[latin1]{inputenc}

\usepackage[latin1]{inputenc}
\usepackage[english]{babel} 
\usepackage{graphicx}
\usepackage{t1enc}

\usepackage[english]{babel} 
\usepackage{graphicx}
\usepackage{t1enc}
 \numberwithin{equation}{section}

\begin{document}
\thispagestyle{empty}
\title{Global Dynamics for the coupled Einstein-Maxwell system with pseudo-tensor of pressure on Bianchi spacetimes.}
\begin{center}
  Global Dynamics for the coupled Einstein-Maxwell system with pseudo-tensor of pressure on Bianchi spacetimes.
\end{center}
\begin{center}
  by
\end{center}

\begin{center}
  Norbert Noutchegueme and Eric Magloire Zangue\\
  Faculty of science, Department of Mathematics, University of
  Yaounde I\\
  P.O.Box:812,Yaounde, Cameroon.
\end{center}

\begin{center}
  \underline{Abstract}
\end{center}

Global existence to the coupled Einstein-Maxwell system which rules
the dynamics of a kind of charged matter with a pseudo-tensor of
pressure is proved, in Bianchi I-VIII spacetimes. We study the
geodesics completeness, the asymptotic behavior, the positivity
conditions, and we prove that the problem is well-posed in the sense
of Hadamard.\\


\section{Introduction}
In relativistic kinetic theory, $global\,\, dynamics$ of several
kinds of charged and uncharged matter remain an active research
domain in General Relativity, in which cosmological constant plays a
central role for astrophysical reasons. In the present paper, we
consider the Einstein equations with in the sources, the
pseudo-tensor of pressure due to A.LICHNEROWICZ [5] which is the
general form of a relativistic fluid tensor. We prove the local and
global existence and we study the asymptotic behavior. We study the
positivity conditions and we prove that the problem is well-posed in
the sense of Hadamard, which means that the global solution is a
continuous
function of the initial data. The work is organized  as follows:\\
 In section 2, we introduce the problem and we give the evolution
 and the constraints systems. In section 3, we study the constraints,
 define the initial values problem and introduce the relative norms.
 In section 4, we construct the local solution by iterated method.
 In section 5, we prove the global existence theorem.
 In section 6, we study the asymptotic behavior and the geodesic
 completeness. In section 7, we study the positivity conditions. In
 section 8, we study the well-posedness in the sense of Hadamard.

 \fancyhead[le]{Global Dynamics for a magnetized plasma} 
\section{The problem.Evolution and constraints systems}
\subsection{The problem}
We are on Bianchi, time oriented spacetimes, which are $(M,\,\,
^4g)$ spacetimes of type I-VIII, where $ ^4g$ is the metric with
signature $(-,+,+,+)$ and $M=\mathbb{R}\times G$ where $G$ is a
three dimensional connected Lie group. $ ^4g$ has the form:
\begin{equation}
^4g= -dt^2 + g_{ij}(t)e^i\otimes e^j
\end{equation}
where $(e_i)$ is a left invariant basis in $G$ and $(e^i)$ the dual
basis; $g=(g_{ij})$ is a positive definite 3-dimensional metric on
$G$ and we adopt the Einstein summation convention $A^\alpha
B_\alpha=\sum\limits_{\alpha} A^\alpha B_\alpha.$ The Greek indexes
$\alpha,\,\beta,...$ range from 0 to 3 and the Latin indexes
$i,\,j,...$ from 1 to 3. The index 4 $(as \,\,^4g)$ is for
quantities on M. The vector $n=\partial_t$ being orthogonal to $G$,
we complete the basis $(e_i)$ on $G$ to obtain a basis $(n,e_i)$ =
$(\partial_0,e_i)$ on M with:
\begin{equation}
\,\,^4e_0 =\,\,\partial_t\,; \,\,^4e_i =\,\,e_i;\,
\,\,^4e_\alpha=\,\,^4e_\alpha^\lambda\partial_\lambda
\end{equation}
where
\begin{equation}
^4e_0^0=\,\,1;\,\,^4e_0^i=\,\,\,\,^4e_i^0=0;\,\,^4e_j^i=\,\,\,\,e_j^i
\end{equation}
The structure constants of the Lie algebra $\mathcal{G}$ of the Lie
group $G$ are denoted $C_{ij}^k$ and defined  by:
\begin{equation}
[e_i,e_j]=C_{ij}^ke_k
\end{equation}
where $[\,\,,\,\,]$ is the Lie brackets in $\mathcal{G}$. Due to the
antisymmetry of $[\,\,,\,\,]$, we have:
\begin{equation}
C_{ij}^k=-C_{ji}^k
\end{equation}
If $^4\nabla$ is the covariant derivative in $^4g$, the Ricci
rotation coefficients $^4\gamma_{\alpha\beta}^\lambda$ are defined
by:
\begin{equation}
^4\nabla_{_{^4e_\alpha}}\,\,^4e_\beta=\,\,^4\gamma_{\alpha\beta}^\lambda\,\,
^4e_\lambda
\end{equation}
If $^4T^\alpha$ and $^4T_{\alpha}$ are tensors on M, we have:
\begin{equation}
\left\{
\begin{array}{lll}
^4\nabla_\alpha\,\,^4T^{\beta}&=&^4e_\alpha(\,\,^4T^{\beta})+\,\,^4\gamma_{\alpha\lambda}^\beta\,\,
\,\, ^4T^{\lambda}\\
^4\nabla_\alpha\,\,^4T_{\beta}&=&^4e_\alpha(\,\,^4T_{\beta})-\,\,^4\gamma_{\alpha\beta}^\lambda\,\,
\,\, ^4T_{\lambda}
\end{array}
\right.
\end{equation}
We study a fluid, subject to the following system in which $t$ is
the only independent variable:

\begin{numcases}
       \strut ^4R_{\alpha\beta}-\frac{1}{2}\,\,^4R\,\,^4g_{\alpha\beta}+
       \Lambda\,\,^4g_{\alpha\beta}=8\pi(^4\tau_{\alpha\beta}
       \,+\,^4T_{\alpha\beta})\label{},\\
 ^4\nabla_{\alpha}\,\,^4F^{\alpha\beta} = e\,\,^4u^{\beta}    \label{},\\
  ^4\nabla_{\alpha}\,\,^4F_{\beta\gamma}+\,\,\,^4\nabla_{\beta}\,\,^4F_{\gamma\alpha}
  \,\,+\,\,^4\nabla_{\gamma}\,\,^4F_{\alpha\beta}\,=\,0; \label{}
\end{numcases}
where:\\
 $\bullet$ $(2.8)$ are the Einstein equations in
 $(\,\,^4g_{\alpha\beta})$ with $^4R_{\alpha\beta}$ the Ricci
 tensor, $^4R=\,\,^4g^{\alpha\beta}\,\,^4R_{\alpha\beta}$ the scalar
 curvature, $\Lambda$ a constant called the cosmological constant,
 $^4\tau_{\alpha\beta}$ the Maxwell tensor associated to the
 electromagnetic field $^4F=(\,\,^4F_{\alpha\beta})$ and defined by:
\begin{equation}
^4\tau_{\alpha\beta}=-
\frac{1}{4}\,\,^4g_{\alpha\beta}\,\,^4F^{\lambda\mu}\,\,^4F_{\lambda\mu}+\,\,^4F_{\alpha\lambda}
\,\,^4F_\beta^{\,\,\lambda}
\end{equation}
and $^4T_{\alpha\beta}$ the symmetric 2-tensor:
\begin{equation}
^4T_{\alpha\beta}=
\frac{4}{3}\rho\,\,^4u_\alpha\,\,^4u_\beta+\,\,^4\Theta_{\alpha\beta}
\end{equation}
where $\rho>0$ is an unknown function of $t$, called the proper
density, \,$^4u=(\,\,^4u^\alpha)$ is the unknown material velocity
of the fluid, $^4u  $ is a unit vector oriented towards the future
direction, and $\,\,^4\Theta_{\alpha\beta}$ is a symmetric 2-tensor
called the $pseudo-tensor\,\, of \,\,pressure$, due to
A.LICHNEROWICZ $[5]$. $\,\,^4\Theta_{\alpha\beta}=0$ is the case
corresponding to pure matter, and
$\,\,^4\Theta_{\alpha\beta}=p\,\,^4g_{\alpha\beta}$ where $p$ is a
scalar function representing the pressure,  corresponds to a perfect
relativistic fluid. We adopt on $\,\,^4\Theta_{\alpha\beta}$ the
assumptions:
\begin{equation}
^4\nabla_{\alpha}\,\,^4\Theta^{\alpha\beta}\,=\,\rho\,\,^4u^{\beta}\,\,\,\,\,and\,\,\,\,
^4\Theta_{0\alpha}=\,\,^4z_0\,\,^4z_\alpha\left(\delta^0_\alpha\,\,+\,\frac{1}{2}\,\delta^{1,2,3}_\alpha\right)
\end{equation}
where $(\,\,^4z_\alpha)$ is a future pointing vector,
$\delta^0_\alpha$  is the Kroneker's symbol and
$\delta^{1,2,3}_\alpha$ is a similar symbol such that
$\delta^{1,2,3}_i=1$ and $\delta^{1,2,3}_0=0$. The first formula in
$(2.13)$ is due to A.LICHNEROWICZ and the second formula which gives
$\,\,^4\Theta_{00}=(\,\,^4z_0)^2$ and hence,
$\,\,^4T_{00}=\frac{4}{3}\rho\,\,(\,\,^4u_0)^2+\,\,^4\Theta_{00}\geq0$
will be very helpful. We suppose that:
\begin{equation}
^4\Theta_{ij}\,\,=\,\,\frac{1}{3}\rho\,\,\,^4g_{ij}
\end{equation}
from where, we have:
\begin{equation*}
^4T_{ij}= \frac{4}{3}\rho\,\,^4u_i\,\,^4u_j+\frac{1}{3}\rho
\,\,^4g_{ij}
\end{equation*}
which is the stress-tensor of a perfect fluid of pure radiation. \\
$\bullet$ (2.9) and (2.10) are the first and second group of the
Maxwell equations, for the electromagnetic field
$\,\,^4F=(\,\,^4F^{0i},\,\,^4F_{ij})$ which is an antisymmetric
 closed 2-form. $\,\,^4F^{0i}$ and $\,\,^4F_{ij}$ are respectively, the
electric and magnetic parts of $\,\,^4F$.\\
In equations (2.9), $e\geq0$ is an unknown function called the
Maxwell current created by the charged particles. \\
(2.10) expresses only the fact that $d\,\,^4F=0$.\\
Let us recall that to solve the system (2.8)-(2.9)-(2.10) is to
determine all the unknown functions $\,\,^4g_{\alpha\beta}$,
$\,\,^4F_{\alpha\beta}$, $\rho$, $\,\,^4u^{\alpha}$,
$\,\,^4\Theta_{\alpha\beta}$ and $e$, which depend only on the time
variable $t$.

\subsection{The equations}

The unit vector $\,\,^4u$ satisfies the relation:
\begin{equation}
^4u^{\alpha}\,\,^4u_\alpha=-1
\end{equation}
which gives, since $\,\,^4u$ is future oriented and using $(2.1)$:
\begin{equation}
^4u^0\,\,=\,\,\sqrt{1\,+\,g_{ij}\,\,^4u^i\,\,^4u^j}.
\end{equation}
We study the Einstein equations in 3+1 formulation, which means
that, they are seen as giving the evolution of the triplet
$(\sum_t,g_t,K_t)$,  where $\Sigma_t=\{t\}\times G$,
$g_t=(g_{ij}(t))$ is called the first fundamental form of
$\Sigma_t$, and $K_t=(K_{ij}(t))$ is called the second fundamental
form of $\Sigma_t$. In the present case, $K_{ij}$ is defined by:
\begin{equation}
K_{ij}=-\frac{1}{2}\partial_tg_{ij}.
\end{equation}
We now introduce a very useful quantity called the mean curvature
and defined by:
\begin{equation}
H= g^{ij}K_{ij}.
\end{equation}
Let us give the expressions of $^4\gamma^\lambda_{\alpha\beta}$ as
they are in [1]. They are:

\begin{equation}
\left\{
\begin{array}{lll}
^4\gamma^\lambda_{00}=\,\,^4\gamma^0_{i0}&=&\,\,^4\gamma^0_{0i}=0;\,\,^4\gamma^0_{ij}
=-K_{ij};\,\,^4\gamma^j_{i0}=\,\,^4\gamma^j_{0i}=-K_{i}^j;\\
^4\gamma^l_{ij}:=\gamma^l_{ij}&=&\frac{1}{2}g^{lk}\left[-C^m_{jk}g_{im}+C^m_{ki}g_{jm}+C^m_{ij}g_{km}\right].
\end{array}
\right.
\end{equation}
We deduce from (2.19) that:
\begin{equation}
\gamma^i_{ij}=C^i_{ij};\,\, \gamma^l_{ij}-\gamma^l_{ji}=C^l_{ij}.
\end{equation}
\subsubsection{\underline{Equations in $\rho$, $\,\,^4u^\alpha$, $\,\,^4\Theta^{0\alpha}$ and $e$ }}
We always have:
\begin{equation*}
^4\nabla_\alpha\left(\,\,^4R^{\alpha\beta}-\frac{1}{2}\,\,^4R\,\,^4g^{\alpha\beta}+\Lambda\,\,^4g^{\alpha\beta}\right)=0.
\end{equation*}
Hence (2.8) implies the conservation conditions:
\begin{equation}
^4\nabla_\alpha\,\,^4\tau^{\alpha\beta}+\,\,^4\nabla_\alpha\,\,^4T^{\alpha\beta}=0
\end{equation}
from where we deduce, using expression (2.12) of
$\,\,^4T_{\alpha\beta}$ and the assumptions (2.13) on
$\,\,^4\Theta^{\alpha\beta}$:
\begin{equation}
^4\nabla_\alpha\,\,^4\tau^{\alpha\beta}+\,\,\frac{4}{3}\,\,\,\,^4\nabla_\alpha(\rho\,\,\,^4u^\alpha\,\,\,^4u^\beta)
+\rho\,\,^4u^\beta=0.
\end{equation}
Now a direct calculation gives:
\begin{numcases}
  \strut
   ^4\nabla_\alpha\,\,^4\tau^{\alpha\beta}=\,\,^4F^\beta_{\,\,\,\,\lambda}\,\, ^4\nabla_\alpha\,\,^4F^{\alpha\lambda}\\
   ^4\nabla_\alpha (\rho\,\,^4u^\alpha\,\,^4u^\beta)=\,\,^4u^\beta\,\, ^4\nabla_\alpha (\rho\,\,^4u^\alpha)+
    (\rho\,\,^4u^\alpha)\,\, ^4\nabla_\alpha (\,\,^4u^\beta).
\end{numcases}

We deduce from (2.21), using, (2.23)-(2.24) and the Maxwell equation
(2.9):
\begin{equation}
e\,\,^4F^\beta_{\,\,\,\,\lambda}\,\,^4u^\lambda +\frac{4}{3}\left[
\,\, ^4u^\beta\,\, ^4\nabla_\alpha (\rho\,\,^4u^\alpha)+
    (\rho\,\,^4u^\alpha)\,\, ^4\nabla_\alpha (\,\,^4u^\beta)\right]+
\rho\,\,^4u^\beta=0.
\end{equation}
The contracted multiplication of (2.25) by $\,\,^4u_{\beta}$ gives,
using
$\,\,^4F^\beta_{\,\,\,\,\lambda}\,\,^4u^\lambda\,\,^4u_{\beta}=
\,\,^4F_{\beta\lambda}\,\,^4u^\lambda\,\,^4u^{\beta}=0$ (since
$\,\,^4F$ is antisymmetric) and $\,\,^4u^\beta\,\,^4u_{\beta}= -1$
which implies $^4u_{\beta}\,\,^4\nabla_\alpha\,\,^4u^\beta\,\,= 0$:

\begin{equation}
\frac{4}{3}\,\, ^4\nabla_\alpha (\rho\,\,^4u^\alpha)+ \rho=0
\end{equation}
If we return to (2.25), we obtain:
\begin{equation}
\frac{4}{3}\,\,(\rho\,\,^4u^\alpha)\,\,^4\nabla_\alpha\,\,^4u^\beta+
e\,\,^4F^\beta_{\,\,\,\,\lambda}\,\,^4u^\lambda=0
\end{equation}
We deduce from (2.26) that:

\begin{numcases}
  \strut
   \dot{\rho}=A(t)\rho\\
   with\,\,\,\,
A(t)=-\left( \frac{3}{4}\,\,\frac{1}{^4u^0}\,
    +\frac{^4\dot{u}^0}{^4u^0}\,-K^i_i+C^i_{ij}\,\frac{^4u^j}{^4u^0}\right)
\end{numcases}
where the dot denotes the derivative with respect to $t$.
Integrating (2.28), we obtain:
\begin{equation*}
\rho=\rho(0)\exp\left(\int_0^tA(s)ds\right)
\end{equation*}
which shows that: $(\rho=0)\Longleftrightarrow(\rho(0)=0)$. In what
follows, we suppose that:
\begin{equation}
\rho(0)>0
\end{equation}
which implies $\rho>0$. (2.27) gives:
\begin{equation}
\,\,^4u^\alpha\,\,^4\nabla_\alpha (\,\,^4u^\beta)=\frac{3}{4}
\frac{e}{\rho}\,\,^{4}F_\lambda^{\,\,\,\beta}\,\,^4u^\lambda
\end{equation}
which is the differential system of current lines. It shows that in
the vacuum $(\,\,^{4}F=0)$, $\,\,^{4}u$ satisfies the geodesics
equation: $\,\,^4u^\alpha\,\,^4\nabla_\alpha (\,\,^4u^\beta)=0$.
Equation (2.31) gives, taking $\beta=0$ and $\beta=i$, using (2.19)
and (2.7):
\begin{numcases}
  \strut
   ^4\dot{u}^0=K_{ij}\,\,\frac{^4u^i\,\,^4u^j}{^4u^0}+ \frac{3}{4}
\frac{e}{\rho}\,\, ^{4}F_i^{\,\,\,0}\,\,\frac{^4u^i}{^4u^0}\\
   ^4\dot{u}^i=2K_j^{i}\,\,^4u^j-\gamma_{jk}^i\,\,\frac{^4u^j\,\,^4u^k}{^4u^0}+ \frac{3}{4}
\frac{e}{\rho}\,\, ^{4}F_0^{\,\,\,i}+ \frac{3}{4} \frac{e}{\rho}\,\,
^{4}F_j^{\,\,\,i}\,\,\frac{^4u^j}{^4u^0}
\end{numcases}
and equation (2.26) in $\rho$ writes:
\begin{equation}
\dot{\rho}=-\left(
\frac{3}{4}\,\,\frac{1}{^4u^0}\,+K_{ij}\frac{^4u^i\,\,^4u^j}{\left(^4u^0\right)^2}
    -K^i_i+C^i_{ij}\,\frac{^4u^j}{^4u^0}\right)\rho
    -\frac{3}{4}
e\,\, ^{4}F_i^{\,\,\,0}\,\,\frac{^4u^i}{\left(^4u^0\right)^2}\,.
\end{equation}
Next, by (2.14) we have: $^4\Theta^{ij}=\frac{\rho}{3}\,\,^4g^{ij}$;
so, $^4\Theta^{ij}$ is given by $\rho$ and $^4g^{ij}$. We then only
look for $^4\Theta^{0\alpha}$. We deduce from (2.13), (2.7), that
$^4\Theta^{0\alpha}$ satisfy the system:
\begin{numcases}
  \strut
   ^4\dot{\Theta}^{00}=K^{i}_{i}\,\,^4\Theta^{00}-C^i_{ij}\,\,^4\Theta^{0j}+\frac{1}{3}\rho
   g^{ij}K_{ij}
+\rho \,\,\,^4u^0;\label{}\\
   ^4\dot{\Theta}^{0i}=K^{j}_{j}\,\,\,\,^4\Theta^{0i}+2K^i_k\,\,^4\Theta^{0k}-\frac{\rho}{3}(C^k_{kj}g^{ij}
+ \gamma^{\,i}_{j\,k}g^{j\,k})+\rho\,\, \,^4u^i.\label{}
\end{numcases}
Now the electromagnetic field always satisfies the identities
$^4\nabla_{\beta}\,\,^4\nabla_{\alpha}\,\,^4F^{\alpha\beta}=0$. This
implies, using the Maxwell equation (2.9):
\begin{equation}
 ^4\nabla_{\alpha}\,\,(e\,\,\,^4u^\alpha)=0.
\end{equation}
(2.37) gives:
\begin{equation}
\dot{e} =
-\left(\frac{1}{^4u^0}\,\,^4\nabla_{\alpha}\,\,^4u^\alpha\right)e
\end{equation}
which integrates to give:
\begin{equation}
e(t)=e(0)\exp\left(-\int_0^t\frac{^4\nabla_\alpha\,\,\,^4u^\alpha}{^4u^0}ds\right).
\end{equation}
This shows that:
\begin{equation*}
(e(0)\geq0)\Leftrightarrow (e\geq0).
\end{equation*}
In what follows we adopt:
\begin{equation}
e_0:=e(0)\geq0.
\end{equation}
Using the equations (2.32)-(2.33) in $\,^4u^\alpha$, equation (2.38)
in $e$ can write:
\begin{equation}
\dot{e}=-\left(\frac{K_{ij}u^iu^j}{(u^0)^2}+C^i_{ik}\frac{^4u^k}{^4u^0}-K^i_i\right)e
  +\frac{3}{4}g_{ij}\frac{^4u^i\,\,^4F^{0j}}{(^4u^0)^2\rho}.
\end{equation}
\subsubsection{\underline{Equations in $\,\,^4F^{0i}$, $\,\,^4F_{ij}$ and constraints } }

The Maxwell equations (2.9) for $\beta=i$, namely
$^4\nabla_{\alpha}\,\,^4F^{\alpha i}=e\,\,^4u^i$, can write using
(2.7) and (2.19), to give the equation for the electric part:
\begin{equation}
^4\dot{F}^{0i}=K_{j}^j\,\,^4F^{0i}-C^j_{jk}\,\,^4F^{ki}
  -\frac{1}{2}\,\,C^i_{jk}\,\,^4F^{jk}-C^j_{jk}\,\,^4F^{0k}\,\,\frac{^4u^i}{^4u^0}.
\end{equation}
Now, the Maxwell equations (2.10) split into the equations:
\begin{numcases}
       \strut
 ^4\nabla_{0}\,\,^4F_{ij}+\,\,^4\nabla_{i}\,\,^4F_{j0}
  \,\,+\,\,^4\nabla_{j}\,\,^4F_{0i}\,=\,0; \label{}\\
  ^4\nabla_{i}\,\,^4F_{jk}+\,\,^4\nabla_{j}\,\,^4F_{ki}
  \,\,+\,\,^4\nabla_{k}\,\,^4F_{ij}\,=\,0. \label{}
\end{numcases}
(2.43) can write using (2.7), and (2.19), to give the equation for
the magnetic part:
\begin{equation}
^4\dot{F}_{ij}= C^k_{ij}g_{kl}\,\,^4F^{0l}.
\end{equation}
Next, for $\beta=0$, the Maxwell equations
$^4\nabla_{\alpha}\,\,^4F^{\alpha 0}=e\,\,^4u^0$ gives the
constraints:

\begin{equation}
 C^i_{ik}\,\,^4F^{0k}+e\,\,^4u^0=0
\end{equation}
and (2.44) gives the constraints:
\begin{equation}
 C^l_{ij}\,\,^4F_{kl}+C^l_{j\,k}\,\,^4F_{il}+C^l_{ki}\,\,^4F_{j\,l}\equiv C^l_{[\,ij}\,\,^4F_{k\,]\,l}=0.
\end{equation}
\subsubsection{\underline{The Einstein equations in 3+1 formulation. Notations.System}}
The principle of the 3+1 formulation is to deduce from the Einstein
equations which are a $second \,\,order$ partial differential
equations system, an equivalent $first\,\, order$ system in $g_{ij}
$ and $K_{ij}$. In the uncharged case, $(^4F=0)$, the calculation is
classical. We did the charged case  $(^4F\neq0)$ in [6] and it is
adopted here without difficulty. So are the constraints. Expressing
the tensor $^4\tau_{\alpha\beta}$ we obtain:
\begin{equation}
\left\{
\begin{array}{lll}
 ^4\tau_{00}=\frac{1}{2}g_{ij}\,\,^4F^{0i}\,\,^4F^{0j}+
  \,\,\frac{1}{4}g^{ik}\,\,g^{jl}\,\,^4F_{kl}\,\,^4F_{ij}\\
  ^4\tau_{0j}= -\,\,^4F^{0k}\,\,^4F_{jk} \\
^4\tau_{ij}=\left(\frac{1}{2}g_{ij}g_{kl}-g_{ik}g_{jl}\right)\,\,^4F^{0k}
\,\,^4F^{0l}-\frac{1}{4}g_{ij}\,\,g^{km}g^{nl}\,\,^4F_{kl}\,\,^4F_{mn}+
g^{kl}\,\,^4F_{ik}\,\,^4F_{jl}
\end{array}
\right.
\end{equation}

For $^4T_{\alpha\beta}$ we have:
\begin{equation}
 ^4T_{00}=\frac{4}{3}\rho(^4u_0)^2+\,\,^4\Theta_{00},\,\,^4T_{0j}=\frac{4}{3}\rho\,\,^4u_0\,\,^4u_j
 +\,\,^4\Theta_{0j},\,\,^4T_{ij}=\frac{4}{3}\rho\,\,^4u_i\,\,^4u_j
 +\frac{1}{3}\rho\,\,g_{ij}.
\end{equation}
Since the indexes are now clearly specified, we will denote in what
follows:
\begin{equation}
\left\{
\begin{array}{lll}
^4F^{0i}=\,\,E^i;\,\,^4F_{ij}=F_{ij};\,\,^4T_{00}=\,\,T_{00};\,\,^4T_{0i}=\,\,T_{0i}
;\,\,^4T_{ij}=\,\,T_{ij}\\
  \,\,^4\tau_{00}=\,\,\tau_{00}; \,\,^4\tau_{0i}=\,\,\tau_{0i}
;\,\,^4\tau_{ij}=\,\,\tau_{ij};\\
\,\,^4u^0=\,\, u^0;\,\,^4u^i=\,\, u^i\\
^4\Theta^{00}=\,\,\Theta^{00};\,\,^4\Theta^{0i}=\,\,\Theta^{0i};\,\,^4\Theta^{ij}=\,\,\Theta^{ij}\\
^4g_{ij}=g_{ij}
\end{array}
\right.
\end{equation}
and we will have for $ F^{ki}$ and $ F_i^{\,\,0}$, using (2.50):
$$\left\{
\begin{array}{lll}
 F^{ki}=g^{kl}g^{im}F_{lm}\\
  \,\,^4F_{i}^{\,\,0}=\,\,^4g_{i\lambda}\,\,^4F^{\lambda0} =\,\,^4g_{ij}\,\,^4F^{j0}=-g_{ij}E^{j}
\end{array}
\right.$$  Now using expression (2.48) of $ ^4\tau_{\alpha\beta}$,
(2.49) of $ ^4T_{\alpha\beta}$, the notations (2.50), equations
(2.33)  of $ ^4u^i$, equation (2.34) of $ \rho$, equations
(2.35)-(2.36) of $ ^4\Theta^{0\alpha}$, equation  (2.41)  of $ e$,
equation (2.42) and (2.45) of $^4F^{0i}$ and $^4F_{ij}$, the
constraints (2.46) and (2.47), reference [6], and the notations
(2.50), we obtain the evolution system (S):

\begin{numcases}{(S):}
\strut
   \dot{g}_{ij}=-2K_{ij};\label{}\\
\dot{K}_{ij}=R_{ij}+HK_{ij}-2K^l_jK_{il}-8\pi(\tau_{ij}+T_{ij})+4\pi
(-T_{00}+g^{lm}T_{lm})g_{ij}-\Lambda g_{ij};\label{}\\
\dot{F}_{ij}=C^k_{ij}g_{k\,l}E^l;\label{}\\
\dot{E}^i=HE^i-C^j_{j\,k}E^k\frac{u^i}{u^0}-C^j_{j\,k}g^{k\,l}g^{im}F_{lm}
-\frac{1}{2}C^i_{j\,k}g^{j\,l}g^{km}F_{lm};\label{}\\
\dot{u}^i=2K^i_ju^j-\gamma^{\,i}_{j\,k}\frac{u^j\,u^k}{u^0}+\frac{3}{4}
\frac{1}{\rho u^0}
C^{j}_{j\,k}E^kE^i-\frac{3}{4}\frac{1}{\rho(u^0)^2}C^{j}_{j\,k}E^kg^{il}
F_{ml}u^m;\label{}\\
\dot{\Theta}^{00}=H\Theta^{00}-C^i_{ij}\Theta^{0j}+\frac{1}{3}\rho H
+\rho u^0;\label{}\\
\dot{\Theta}^{0i}=H\Theta^{0i}+2K^i_j\Theta^{0j}-\frac{\rho}{3}(C^k_{kj}g^{ij}
+ \gamma^{\,i}_{j\,k}g^{j\,k})+\rho u^i;\label{}\\
\dot{\rho}=-(\frac{3}{4u^0}+K_{ij}\frac{u^i\,u^j}{(u^0)^2}-K^i_i+C^i_{ij}
\frac{u^j}{u^0})\rho-\frac{3}{4}g_{il}C^j_{j\,k}E^kE^l\frac{u^i}{(u^0)^3};\label{}\\
\dot{e}=-\left(\frac{K_{ij}u^iu^j}{(u^0)^2}+\frac{C^i_{ik}u^k}{u^0}-K^i_i\right)e
  +\frac{3}{4}g_{ij}\frac{u^iE^j}{(u^0)^2}\frac{e^2}{\rho};\label{}
\end{numcases}

Where in (2.52), $R_{ij}$ is the Ricci tensor associated to $g_{ij}$
and whose expression due to R.T. JANTZEN [3] is:
\begin{equation}
R_{ij}=\gamma^l_{lm}\gamma^m_{ji}-\gamma^m_{jl}\gamma^l_{mi}-C^l_{mj}\gamma^m_{li}.
\end{equation}
Next, we obtain the constraints:
\begin{numcases}
  \strut
   R-K_{ij}K^{ij}+H^2=16\pi(\tau_{00}+T_{00})+2\Lambda,\\
   \nabla^iK_{ij}=-8\pi(\tau_{0j}+T_{0j}),\\
  C^l_{ij}F_{kl}+C^l_{j\,k}F_{il}+C^l_{ki}F_{j\,l}\equiv C^l_{[\,ij}F_{k\,]\,l}=0,\\
  C^i_{ik}F^{0k}+eu^0=0,
\end{numcases}
where $R=g^{ij}R_{ij} $ and  $\nabla$ is the Levi-Civita connection
associated to $g=(g_{ij})$.  The constraint (2.61) is called the
Hamiltonian constraint. As we will see, this constraint is
fundamental.

\section{\underline{Study of constraints.The initial values problem.Relative}}
\section*{\,\,\,\,\,\,\,\,\,\,\,\,\,\,\,\,\,\,\,\, \,\,\,\,\,\,\,\,\,\,\,\,\,\,\,\,\,\,\,\,
\,\,\,\,\,\,\,\,\,\,\,\,\,\,\,\,\,\,\,\,\,\,\underline{Norms}}
\subsection{\underline{Study of constraints}}

We prove that, if the constraints are satisfied by the solutions of
the evolution system at $t=0$, then the constraints are satisfied
everywhere. Let us set:

\begin{center}
$$\left\{
\begin{array}{lll}
   A&=&R-K_{ij}K^{ij}+H^2-16\pi(\tau_{00}+T_{00})-2\Lambda,\\
   A_j&=&\nabla^iK_{ij}+8\pi(\tau_{0j}+T_{0j}),\\
 A_{ijk}&=&C^l_{[\,ij}F_{k\,]\,l},\\
  B&=&C^i_{ik}F^{0k}+eu^0,\\
W&=&(A,A_j,A_{ijk},B)^t.
\end{array}
\right.$$
\end{center}

\begin{proposition}
1) The quantities $A\,\,,A_j,\,A_{ijk},\,B$ satisfy the relations:
\begin{numcases}
  \strut
   \dot{A}=2HA+2g^{ij}\gamma^k_{ij}A_k,\label{}\\
 \dot{A}_j=HA_j,\label{}\\
 \dot{A}_{ijk}=0,\label{}\\
 \dot{B}=HB.\label{}
\end{numcases}
2) The solutions of the evolution system satisfy the constraints
everywhere, if and only if they satisfy the constraints at $t=0$.
\end{proposition}

\begin{proof}
\begin{itemize}
\item[$1)$]
 See Appendix .
\item[$2)$] If we set $ W=(A,\,A_j,\,A_{ijk},\, B)^t$, then $ W $
satisfies: $ \Dot{W}=L W$,  where $ L $ is a matrix. So we have:
\begin{equation*}
W(t)=W(0)\exp\left(\int_0^tL(s)ds\right).
\end{equation*}
This shows that $(W=0)\Leftrightarrow (W(0)=0)$. This means that the
constraints are satisfied everywhere if and only if they are
satisfied at $t=0$. In what follows, we consider that the
constraints are properties of the solutions of the evolution system.
\end{itemize}
\end{proof}

\subsection{\underline{The initial values problem}}

Let: $g^0=(g^0_{ij})$, $K^0=(K^0_{ij})$, $F^0=(F^0_{ij})$ be given
$3\times3$ matrices with $g^0$ positive definite, $K^0$ symmetric
and $F^0$ antisymmetric matrices.\\
 Let: $U^0=(U^{0,i})$,
$E^0=(E^{0,i})$, $\theta^0=(\theta^{0,\alpha})$ be given vectors.\\
Let: $\rho^0>0$, and $e^0\geq0$ be given numbers.\\
We look for solutions: $g =(g_{ij})$ positive definite $3\times3$
matrix, $K =(K_{ij})$ symmetric $3\times3$ matrix, $F =(F_{ij})$
antisymmetric $3\times3$ matrix, $u =(u^i)$, $E =(E^{i})$, $\Theta
=(\Theta^{0\alpha})$  vectors, $\rho >0$, and $e\geq0$ two
functions,  such that at $t=0$, we have:

\begin{center}
$$\left\{
\begin{array}{lllllllll}
 g(0)&=&g^0;\,\,\,\, K(0)&=&K^0;\,\,\,\,
F(0)&=&F^0;\,\,\,\, E(0)&=&E^0;\\
u (0)&=&U^0;\,\,\,\, \Theta(0)&=&\theta^0;\,\,\,\, \rho
(0)&=&\rho^0;\,\,\,\, e (0)&=&e^0;
\end{array}
\right.$$
\end{center}

For $t\in [0,T[,\,T\leq+\infty$. We will set:

\begin{equation}
U^{0,0}=\sqrt{1+g^0_{ij}U^{0,i}\,U^{0,j}};\,\,\,\, u^{ 0}=\sqrt{1+g
_{ij}u^{i} u^{j}}
\end{equation}

$g^0$, $K^0$, $F^0$, $E^0$, $U^0$, $\theta^0$, $\rho^0$ and $e^0$
are the initial data. In what follows, we suppose that they satisfy
the constraints at $t=0$. There are eight initial data for the four
constraints. This means that there are four degrees of liberty to
the initial data.
\subsection{Relative norms}
We now introduce the notion of relative norms due to RENDALL $[7]$.

\begin{lemma}
Define the norm of the $n\times n$ matrix $A$ by:
\begin{center}
$||A||=sup\left\{\frac{||Ax||}{||x||},x\neq0,x\in\mathbb{R}^n\right\}.$
\end{center}
If $A_1$ and $A_2$ are two $n\times n$ matrices, $A_1$ positive
definite, the norm of $A_2$ with respect to $A_1$ is:
\begin{center}
$||A_2||_{A_1}=sup\left\{\frac{||A_2x||}{||A_1x||},x\neq0,x\in\mathbb{R}^n\right\}.$
\end{center}
\end{lemma}
From the definition we have:
\begin{equation}
||A_2||\leq  ||A_2||_{A_1}  ||A_1||
\end{equation}
We also have:
\begin{equation}
||A_2||_{A_1}\leq
\left[Tr(A^{-1}_1A_2A^{-1}_1A_2)\right]^{\frac{1}{2}}
\end{equation}
If $A=(a_{ij})$ is a  $n\times m$ matrix, one defines another norm
by:
\begin{equation}
|A|=sup\left\{|a_{ij}|;i=1,2,...,n; j=1,2,...,m\right\}.
\end{equation}

\begin{lemma}
Let $(u^\alpha)=(u^0,u^i)$, where $u^0$ and $u^i$ are linked by
(3.5). Let $F^{0i}$ be given.   Then, there exists  a constant $C>0$
such that:
\begin{equation}
\left| \frac{u^i}{u^0}\right|\leq C\left|g\right|^{\frac{3}{2}};\,\,
\left|F^{0i}\right|\leq\left(g_{rs}F^{0r}F^{0s}\right)^\frac{1}{2}\left|g\right|^\frac{3}{2}.
\end{equation}
\end{lemma}

\begin{proof}
Take in lemma 3.1, $A_1=(g^{ij})$, $A_2=(a^{ij})$ with
$a^{ii}=u^{i}$ and $ a^{ij}=0$ if $i\neq j$. A direct calculation,
using $A_1^{-1}=(g_{ij})$ gives:
\begin{equation}
Tr(A^{-1}_1A_2A^{-1}_1A_2)=a^{ij}a_{ij}
\end{equation}
(3.6) and (3.7) then give, using (3.10):
\begin{equation}
||A_2||\leq\left(a^{ii}a_{ii}\right)^\frac{1}{2}||g||.
\end{equation}
But we have:
\begin{center}
 $\begin{array}{lll}
(a^{ij}a_{ij})^\frac{1}{2}&=&(a^{ii}a_{ii})^\frac{1}{2}\\
&=&(g_{ik}g_{il}a^{kl}a^{ii})^\frac{1}{2}\\
&=&(g^2_{ik}a^{kk}a^{ii})^\frac{1}{2}\leq
\left|g\right|^\frac{1}{2}\left(g_{ik}u^iu^k\right)^\frac{1}{2}.
\end{array}$
\end{center}
The first relation (3.9) is due to (3.11) using the inequalities
$\left|u^i\right|\leq |A_2|\leq C||A_2||$, since all the norms on a
finite dimensional vector space are equivalent, the relation
$u^0=\left(1+ g_{ik}u^iu^k\right)^\frac{1}{2}$ and the fact that $
|g|$ and $||g||$ are equivalent. For the second relation (3.9), take
$a^{ii}=F^{0i}$ and $ a^{ij}=0$ if $i\neq j$ and proceed as above.
This ends the proof of lemma 3.2.
\end{proof}

\section{\underline{Local existence by iterated method}}

\subsection{\underline{Construction of the iterated sequence}}

Consider the initial values in paragraph   3.2.  Set:

\begin{center}
$$\left\{
\begin{array}{lllllllll}
 g_0(t)&=&g^0;\,\,\,\, K_0(t)&=&K^0;\,\,\,\,
F_0(t)&=&F^0;\,\,\,\, E_0(t)&=&E^0;\\
u_0(t)&=&U^0;\,\,\,\, \Theta_0(t)&=&\theta^0;\,\,\,\,
\rho_0(t)&=&\rho_0;\,\,\,\, e_0(t)&=&e_0;
\end{array}
\right.$$
\end{center}

We also set: $V_0=(g_0,K_0,F_0,E_0,u_0,\Theta_0,\rho_0,e_0)$. For
$n\in\mathbb{N}$, if $V_n=(g_n,K_n,F_n,E_n,u_n,\Theta_n,\rho_n,e_n)$
is known, then define  $V_{n+1}$ by:
\begin{equation*}
V_{n+1}= V_0+\int_0^tf(V_n(s))ds
\end{equation*}
where $f(V_n)$ is the right hand side of the evolution system (S) in
which:
\begin{itemize}
\item[$*)$] $u^0_n=\sqrt{1+g_{n,ij}u^i_nu_n^j}\geq1$;
\item[$*)$] $\gamma^k_{n,ij}$ is obtained by replacing in $(2.19)$ $g_{ij}$
by $g_{n,ij}$;
\item[$*)$] $R_{n,ij}$,\,\, $\tau_{n,\alpha\beta}$,\,\,$T_{n,\alpha\beta}$, $\Theta_{n,ij}$
 are defined by the same method.
\end{itemize}
We then obtain the iterated sequence $(V_n)$ which is of class $C^1$
on a maximal interval $[0,T_n[$, $T_n>0$ where $g_n=(g_{n,ij})$ is
symmetric and positive definite, $K_n=(K_{n,ij})$ is symmetric,
$F_n=(F_{n,ij})$ is antisymmetric, $\rho_n>0$ and $e_n\geq0$.

\subsection{Estimation of the iterated sequence}
\begin{proposition}
There exists a number $T>0$, independent of $n$, such that, the
iterated sequence $V_n=(g_n,K_n,F_n,E_n,u_n,\Theta_n,\rho_n,e_n)$ is
defined and uniformly bounded on $[0,T[$.
\end{proposition}

\begin{proof}
Let $N\in\mathbb{N}^{\star}$. Suppose that for $n\leq N-1$, we have
the following inequalities:

\begin{equation}
\left\{
\begin{array}{lllllll}
|g_n-g^0|&\leq & A_1;\,\,\,\, |K_n-K^0|&\leq &A_2;\,\,\,\,\,\,|F_n-F^0|&\leq &A_3;\\
|E_n-E^0|&\leq & A_4;\,\,\,\, |u_n-U^0| &\leq & A_5;\,\,\,\, |\Theta_n-\theta^0| &\leq& A_6;\\
|e_n-e_0| &\leq& A_7;\,\,\,\, |\rho_n-\rho_0| & \leq & A_8;\,\,\,\,
(detg_n)^{-1} &\leq& A_9
\end{array}
\right.
\end{equation}

where $A_i,i=1,2,3,4,5,6,7,8,9,$ are strictly positive constants.
 We prove that we can choose these constants such that we still
 have the inequalities $(4.1)$   for $n=N$, for $T$ sufficiently
 small.\\
 Note that for $t$ sufficiently small, we have
 \begin{center}
$|\rho_{N-1}-\rho^0|\leq \frac{\rho^0}{2}$
\end{center}
 from where we deduce that
\begin{center}
$ \frac{1}{\rho_{N-1}}\leq\frac{2}{\rho^0}.$
\end{center}
we take $A_8=\frac{\rho^0}{2}$. Taking the inequalities $(4.1)$ into
account, the definition of the iterated sequence, the expressions of
\begin{equation*}
H_{N-1},\,\, R_{N-1,ij},\,\, T_{N-1,ij},\,\, \tau_{N-1,ij}\,\, and
\,\,u^0_{N-1}=\sqrt{1+g_{N-1,ij}u^i_{N-1}u^j_{N-1}} \geq 1,
\end{equation*}
we can find constants $B_i, i=1,2,3,4,5,6,7,8$ strictly positive
depending only on $A_i$ such that:

\begin{equation}
\left\{
\begin{array}{lllllllll}
|\dot{g}_{N,ij}|&\leq & B_1;\,\,\,\, |\dot{K}_{N,ij}|&\leq &B_2;\,\,\,\,\,\,|\dot{F}_{N,ij}|&\leq &B_3; \,\,\,\, |\dot{E}_N^i| &\leq& B_4; \\
|\dot{u^i}_N|&\leq & B_5;\,\,\,\, |\dot{\Theta}^{0\alpha}_N| &\leq &
B_6;  \,\,\,\,|\dot{e}_N| &\leq& B_7;\,\,\,\, |\dot{\rho}_N|&\leq&
B_8.
\end{array}
\right.
\end{equation}

By integration of $(4.2)$ we have:
\begin{equation}
\left\{
\begin{array}{lllllllll}
|g_{N}-g^0|&\leq & B_1t;\,\,\,\, |K_N-K^0|&\leq &B_2t;\,\,\,\,\,\,|F_N-F^0|&\leq &B_3t; \,\,\,\, |E_N-E^0| &\leq& B_4t; \\
|u_N-U^0|&\leq & B_5t;\,\,\,\, |\Theta_N-\theta^0| &\leq & B_6t;
\,\,\,\,|e_N-e_0| &\leq& B_7t;\,\,\,\, |\rho_N-\rho_0|&\leq& B_8t.
\end{array}
\right.
\end{equation}

Let us bound $(\det g_N)^{-1}$. By the definition of the iterated
sequence we have:
\begin{equation}
\frac{d}{dt}g_{N,ij}=-2K_{N-1,ij}.
\end{equation}

On one hand we have the equality:
\begin{equation*}
\frac{d}{dt}[\ln(\det g_N)]=g_N^{ij}\frac{d}{dt}g_{N,ij}.
\end{equation*}
On the other hand we have:
\begin{equation*}
\begin{array}{lll}
\frac{d}{dt}[\ln(\det g_N)]&=& (\det g_N)^{-1}\frac{d}{dt}(\det
g_N)\\
&=&-\det g_N\frac{d}{dt}(\det g_N)^{-1}.
\end{array}
\end{equation*}

The relation $(4.4)$ then implies:
\begin{equation*}
\frac{d}{dt}(\det g_N)^{-1}=(2g_N^{ij}K_{N-1,ij})(\det g_N)^{-1}
\end{equation*}

which is a differential equation in $(\det g_N)^{-1}$ on $[0;t]$,
$t>0$, whose solution is:
\begin{equation}
(\det g_N)^{-1}= (\det
g^0)^{-1}\exp\left(\int_0^t2g_N^{ij}K_{N-1,ij}(s)ds\right).
\end{equation}
The relation $(4.4)$ which is similar to  $(2.17)$, shows that $g_N$
and $K_{N-1}$ are the first and second fundamental forms of a
hypersurface. We must then have:
\begin{equation*}
g_N^{ij}K_{N-1,ij}=K^i_{N-1,i}=Tr(K_{N-1}).
\end{equation*}
We then deduce from  $(4.1)$ and $(4.5)$ that there exists a
constant $C>0$ depending on $A_i$, $i=1,2,3,4,5,6,7,8,$ and $K^0$
such that:
\begin{equation}
(\det g_N)^{-1}\leq (\det g^0)^{-1}\exp(Ct),
\end{equation}

Now if we take in  $(4.1)$:
\begin{equation*}
(det g^0)A_9>1,
\end{equation*}
we will have for $ t$ sufficiently small:
\begin{equation*}
(\det g^0)A_9>\exp(Ct)>1.
\end{equation*}
This means, from $(4.6)$ that,  there exists $t_1>0$, such that, for
$0<t\leq t_1$, we have:
\begin{equation}
(\det g_N)^{-1}\leq A_9.
\end{equation}
Then, using $(4.3)$ and $(4.7)$, we conclude that if $T>0$ is such
that:
\begin{equation*}
0<T<t_1,\,\,\,\, B_iT<A_i; i=1,2,3,4,5,6,7,8
\end{equation*}
then we still have the inequalities $(4.1)$ for $n=N$. Hence the
iterated sequence $(V_n)$ is uniformly bounded on $[0,T[$.
\end{proof}

\subsection{Local existence}

\begin{proposition}
The initial values problem for the evolution system
$(2.51)...(2.59)$ has a unique local solution.
\end{proposition}

\begin{proof}
Let $[0,T[$, $T>0$, be the interval  obtained in proposition $4.1$.
We prove that the iterated sequence $V_n=(g_n,K_n,F_n,E_n, u_n,
\Theta_n,\rho_n,e_n)$ converges uniformly on every
$[0,\delta[\subset [0,T[$, $\delta>0$, towards a solution
$V=(g,K,F,E,u,\Theta,\rho,e)$ of the evolution system.\\

We take the difference between two consecutive terms of $V_n$, we
use the fact that they have the same initial data and since the
sequences $\left(\frac{1}{\rho_n}\right)_n$, $(V_n)$ are uniformly
bounded we have, with $C>0$ a constant:
\begin{equation}
\begin{array}{lll}
|g_{n+1}(t)-g_n(t)|+|K_{n+1}(t)-K_n(t)|&+&|F_{n+1}(t)-F_n(t)|\\
+|E_{n+1}(t)-E_n(t)|+|u_{n+1}(t)-u_n(t)|&+&|\Theta_{n+1}(t)-\Theta_n(t)|\\
+|\rho_{n+1}(t)-\rho_n(t)|+|e_{n+1}(t)-e_n(t)|&\leq & C\int_0^t[|g_n(s)-g_{n-1}(s)|\\
+|K_n(s)-K_{n-1}(s)|+|F_n(s)-F_{n-1}(s)|&+&|E_n(s)-E_{n-1}(s)|\\
+|u_n(s)-u_{n-1}(s)|+
|\Theta_n(s)-\Theta_{n-1}(s)|&+&|\rho_n(s)-\rho_{n-1}(s)|+|e_n(s)-e_{n-1}(s)|]ds.
\end{array}
\end{equation}
For the same reasons we have:

\begin{equation}
\begin{array}{lll}
|\frac{dg_{n+1}}{dt}-\frac{dg_n}{dt}|+|\frac{dK_{n+1}}{dt}-\frac{dK_n}{dt}|&+&|\frac{dF_{n+1}}{dt}-\frac{dF_n}{dt}|\\
+ |\frac{dE_{n+1}}{dt}-\frac{dE_n}{dt}|+|\frac{du_{n+1}}{dt}-\frac{du_n}{dt}|&+&|\frac{d\Theta_{n+1}}{dt}
-\frac{d\Theta_n}{dt}| \\
+|\frac{d\rho_{n+1}}{dt}-\frac{d\rho_n}{dt}|+|\frac{de_{n+1}}{dt}-\frac{de_n}{dt}|
&\leq &C[ |g_n-g_{n-1}|\\
+|K_n-K_{n-1}|+|F_n-F_{n-1}|&+&|E_n-E_{n-1}|\\
+|u_n-u_{n-1}|+|\Theta_n-\Theta_{n-1}|&+&|\rho_n-\rho_{n-1}|+
|e_n-e_{n-1}|].
\end{array}
\end{equation}

For $n\in\mathbb{N}$, we set:
\begin{equation}
\begin{array}{lll}
\alpha_n(t)&=&|g_{n+1}(t)-g_n(t)|+|K_{n+1}(t)-K_n(t)|+|F_{n+1}(t)-F_n(t)|\\
&+&|E_{n+1}(t)-E_n(t)|+|u_{n+1}(t)-u_n(t)|+|\Theta_{n+1}(t)-\Theta_n(t)|\\
&+&|\rho_{n+1}(t)-\rho_n(t)|+|e_{n+1}(t)-e_n(t)|
\end{array}
\end{equation}
 $(4.8)$ and $(4.10)$ give:

\begin{equation}
\alpha_n(t)\leq C\int_0^t\alpha_{n-1}(s)ds
\end{equation}
 By induction on  $n\geq 2$, we obtain, from$(4.11)$ :

\begin{equation}
|\alpha_n(t)|\leq ||\alpha_2||_{\infty}\frac{(Ct)^{n-2}}{(n-2)!}
\leq ||\alpha_2||_{\infty}\frac{(C\delta)^{n-2}}{(n-2)!},
\end{equation}

for $0\leq t\leq\delta$ and $0<\delta<T.$\\
But the series $\sum\frac{(C\delta)^{n}}{(n)!}$ converges. Hence we
obtain from $(4.12)$ that:
\begin{equation*}
\lim_{n\rightarrow+\infty}\sup_{0\leq t\leq\delta}\alpha_n(t)=0.
\end{equation*}
Given the definition $(4.10)$ of $\alpha_n$, we conclude that every
sequence $(g_n)$, $(K_n)$, $(F_n)$, $(E_n)$, $(u_n)$, $(\Theta_n)$
$(\rho_n)$ and $(e_n)$ converges uniformly on every interval
$[0,\delta]$, $0<\delta<T$ and we denote the different limits by
 $g$, $K$, $F$, $E$, $u$, $\Theta$, $\rho$ and $e$ which are continuous functions
 of $t$.\\
Now from the inequality $(4.9)$, we conclude similarly that the
sequences of derivatives $(\frac{dg_n}{dt})$, $(\frac{dK_n}{dt})$,
$(\frac{dF_n}{dt})$, $(\frac{dE_n}{dt})$, $(\frac{du_n}{dt})$,
$(\frac{d\Theta_n}{dt})$ $(\frac{d\rho_n}{dt})$, $(\frac{de_n}{dt})$
converge  uniformly on $[0,\delta]$, $0<\delta<T$. In these
conditions, the functions $g$, $K$, $F$, $E$, $u$, $\Theta$, $\rho$
and $e$ are of class $C^1$ on $[0,T[$. Hence
$V=(g,K,F,E,u,\Theta,\rho,e)$ is a local solution of the evolution system $(2.51)...(2.59)$.\\
We now prove that the solution is unique.\\
 Consider two solutions
$V_1$ and $V_2$ of the same initial values problem. Define
$\alpha(t)=|V_1(t)-V_2(t)|$ with $\alpha(0)=0$. Since the functions
$g$, $K$, $F$, $E$, $u$, $\Theta$, $\rho$, $e$, $(\det g)^{-1}$ and
$\frac{1}{\rho}$ are bounded on $[0,\delta]$, $0<\delta<T$, there
exists a constant $C>0$ such that:
\begin{equation*}
\alpha(t)\leq C\int_0^t\alpha(s)ds.
\end{equation*}
By Gronwall Lemma, we obtain $\alpha(t)=0$ since $\alpha(0)=0$. So
$V_1=V_2$ and the local solution is unique.
\end{proof}

\section{Global existence theorem}

We have to prove that, the solution $V=(g,K,F,E,u,\Theta,\rho,e)$
and the functions $\frac{1}{\rho}$ and   $(\det g)^{-1}$ are bounded
on every bounded interval. First of all, we prove the following
important result on the mean curvature $H=g^{ij}K_{ij}$.
\begin{proposition}
Let the function $H=g^{ij}K_{ij}$ be bounded on $[0,T^\star[$ where
$T^\star<+\infty$. Then the functions $g,K,F,E,u,\Theta,\rho$, $e$,
$(\det g)^{-1}$, $u^0=\sqrt{1+g_{ij}u^iu^j}$ and $\frac{1}{\rho}$
are bounded on $[0,T^\star[$.
\end{proposition}
\begin{proof}
We will use the following Lemma:
\begin{lemma}
The mean curvature $H=g^{ij}K_{ij}$ satisfies the following
relation:
\begin{equation}
\frac{dH}{dt}=R+H^2+ 4\pi
g^{ij}(\tau_{ij}+T_{ij})-12\pi(\tau_{00}+T_{00})-3\Lambda.
\end{equation}
\end{lemma}

\begin{center}
\underline{Proof of the Lemma}
\end{center}
Since $H=g^{ij}K_{ij}$, the relation (2.17) gives:
$\dot{g}^{ij}=2K^{ij}$, from there we have:
\begin{equation*}
\frac{dH}{dt}=2K_{ij}K^{ij}+g^{ij}\dot{K}_{ij}
\end{equation*}
then obtain  $(5.1)$ from equation $(2.52)$ by a direct calculation
using $^4g^{\alpha\beta}\,^4\tau_{\alpha\beta}=0$, which implies
$g^{ij}\tau_{ij}=\tau_{0\,0}$.\\
\begin{center}
\underline{Proof of Proposition 5.1}
\end{center}
$\bullet$ \underline{Boundedness of $|g|$ on $[0,T^\star[$}.\\
By using the Hamiltonian constraint $(2.61)$, $(5.1)$ gives:
\begin{equation}
\frac{dH}{dt}=K_{ij}K^{ij}-\Lambda+4\pi g^{ij}(\tau_{ij}+T_{ij})
+4\pi(\tau_{00}+T_{00}).
\end{equation}
But we have  $\tau_{00}=g^{ij}\tau_{ij}\geq0$; $T_{00}\geq0$;
$g^{ij}T_{ij}\geq0$. We then deduce from $(5.2)$:
\begin{equation}
\frac{dH}{dt}\geq K_{ij}K^{ij}-\Lambda.
\end{equation}
 Integrating $(5.3)$ on $[0,t]$ for $0\leq t \leq T^\star$ gives:
 \begin{equation}
H(t)\geq H(0)-\Lambda t +\int_0^tK_{ij}K^{ij}ds.
\end{equation}
which shows, since $H$ is bounded on $[0,T^\star[$ that we have:
\begin{equation}
\int_0^{T^*}K_{ij}K^{ij}ds<+\infty.
\end{equation}
Now the integration of equation $(2.51)$ on $[0;t]$, $0\leq t\leq
T^\star$, gives:
\begin{equation}
|g(t)|\leq |g(0)|+2\int_0^t|K(s)|ds.
\end{equation}

$(5.6)$ gives, using $(3.6)$, the inequality:
\begin{equation}
||g(t)||\leq C\left[||g(0)||+
\int_0^t||K(s)||_{g(s)}||g(s)||ds\right].
\end{equation}
 $(5.7)$, gives, using $(3.7)$:
\begin{equation*}
||g(t)||\leq C\left[
||g(0)||+\int_0^t(K_{ij}K^{ij})^{\frac{1}{2}}||g(s)||ds\right],
\end{equation*}
By  Gronwall Lemma, this implies:
\begin{equation*}
||g(t)||\leq
C_1||g(0)||\exp\left(C\int_0^t(K_{ij}K^{ij})^\frac{1}{2}ds\right),
0\leq t\leq T^\star;
\end{equation*}
This shows, using $(5.5)$ that $||g||$, and therefore
$|g|$, are bounded on $[0,T^\star[$.\\
$\bullet$ \underline{Boundedness of $(\det g)^{-1}$ on $[0;T^\star[$}\\
The relation :
\begin{equation*}
\frac{d}{dt}[\ln(\det g)]=g^{ij}\frac{dg_{ij}}{dt}
\end{equation*}
gives, using equation $(2.52)$ and $H=g^{ij}K_{ij}$:
\begin{equation*}
\frac{d}{dt}[\ln(\det g)]=-2H .
\end{equation*}
But $|H|$ is bounded on $[0,T^\star[$. Integrating on $[0,t]$,
$0\leq t<T^\star[$, there exists a constant $C>0$ such that:
\begin{equation*}
-C<\ln(\det g)<C
\end{equation*}
from where we have:
\begin{equation*}
e^{-C}<\det g<e^{C}.
\end{equation*}
Then $\det g$ and $(\det g)^{-1}$ are bounded on $[0,T^\star[$.\\

$\bullet$\underline{ Boundedness of $|K|$ on $[0,T^\star[$}\\
Since $|g|$ and $(\det g)^{-1}$ are bounded on $[0,T^\star[$, the
expression $(2.60)$ of $R_{ij}$ shows that $R=g^{ij}R_{ij}$ is
bounded on $[0,T^\star[$. We deduce from the Hamiltonian constraint
$(2.61)$, since $\tau_{00}\geq0$, $T_{00}\geq0$, that:
\begin{equation*}
R+H^2-2\Lambda\geq K_{ij}K^{ij}.
\end{equation*}
$R$ and $H^2$ being bounded, so is $K_{ij}K^{ij}$. We then deduce
from the inequality:
\begin{equation*}
||K(s)||\leq (K_{ij}K^{ij})^\frac{1}{2}||g(s)||,
\end{equation*}
that $||K||$ and hence $|K|$ is bounded on $[0,T^\star[$, since
$(K_{ij}K^{ij})^\frac{1}{2}$
and $||g||$ are bounded. \\
$\bullet$ \underline{Boundedness of $E$ on $[0,T^\star[$}\\
We use the Hamiltonian  constraint $(2.61)$ and $\tau_{00}\geq0$,
$T_{00}\geq0$, to have:
\begin{equation*}
0\leq \max(16\pi\tau_{00}; 16\pi T_{00})\leq 16\pi(\tau_{00}+
T_{00})=R+H^2-K_{ij}K^{ij}-2\Lambda.
\end{equation*}
This shows that $\tau_{00}$ and $T_{00}$ are bounded on
$[0,T^\star[$. We deduce from $(2.11)$ that:
\begin{equation*}
g^{ij}\tau_{ij}=\frac{1}{4}F^{ij}F_{ij}+\frac{1}{2}g_{ij}E^iE^j,
\end{equation*}
thus
\begin{equation*}
0\leq\frac{1}{2}g_{ij}E^iE^j\leq \tau_{00},
\end{equation*}
since $F^{ij}F_{ij}\geq0$. We use $(3.9)$ to conclude that  $E$ is
bounded, since
$\left(g_{rs}F^{0r}F^{0s}\right)^\frac{1}{2}$ and $|g|$ are bounded.\\

$\bullet$ \underline{Boundedness of $F$ on  $[0,T^\star[$ }\\
Equation $(2.53)$ gives, since $|g|$ and $|E|$ are bounded:
\begin{equation*}
|\dot{F}_{ij}|\leq C, \forall t\in [0,T^\star[.
\end{equation*}
where $C$ is a constant. Integrating on  $[0,t]$, $0\leq\, t\leq T^\star$ gives the result.\\

$\bullet$ \underline{Boundedness of  $\rho$ on  $[0,T^\star[$ }\\
Since $|E|$, $|g|$, $|K|$, $\frac{u^i}{u^0}$, $\frac{1}{u^0}$ are
bounded, equation  $(2.58)$ in $\rho$  shows that, there exists two
constants $A>0$, $B>0$ such that:
\begin{equation*}
\rho(t)\leq A\int_0^t\rho(s)ds+B,\,\,\,\,\forall t\in [0,T^\star[.
\end{equation*}
By Gronwall Lemma, we have:
\begin{equation*}
\rho(t)\leq Be^{At}\leq Be^{AT^\star},\,\,\,\, \forall t\in
[0,T^\star[
\end{equation*}
 and $\rho$ is bounded.\\

$\bullet$ \underline{Boundedness of $u^0$, $\frac{1}{\rho}$ and $u^i$ on  $[0,T^\star[$  }\\

From $(2.28)$ and $(2.29)$, we have:
\begin{equation*}
\frac{\dot{\rho}}{\rho}+\frac{\dot{u}^0}{u^0}=-(\frac{3}{4}.\frac{1}{u^0}-K^i_i+C^i_{ij}\frac{u^j}{u^0}).
\end{equation*}
Since $\frac{1}{u^0}$,   $H=K^i_i$ and $\frac{u^j}{u^0}$ are bounded
on $[0;T^\star[$, there exists a constant $C>0$ such that:
\begin{equation*}
-C<
-\int_0^t(\frac{3}{4}.\frac{1}{u^0}-K^i_i+C^i_{ij}\frac{u^j}{u^0})(s)ds<C,
\end{equation*}
for $0<t<T^\star$. From where we deduce by integration:
\begin{equation*}
e^{-C}\leq \rho u^0\leq e^{C}\,\,on\,\,[0,T^\star[,
\end{equation*}
which proves that: $\rho u^0$ and $\frac{1}{\rho u^0}$ are bounded
on $[0,T^\star[$. Now write, using $(3.9)$:\\
 $\left|\frac{u^i}{u^0}\right|= \frac{\rho|u^i|}{\rho u^0}\leq C|g|^\frac{3}{2}$ from
where we have:
\begin{equation*}
\rho|u^i|\leq c(\rho u^0)|g|^\frac{3}{2}
\end{equation*}
and this prove that $\rho u^i$ is bounded. Now multiply equation
$(2.55)$ in $ u^i$  by $\rho$, and use the fact that $\rho u^i$ is
bounded to conclude that $\rho \dot{u}^i$ is also bounded. The
relation
\begin{equation*}
u^0=\sqrt{1+g_{ij}u^iu^j}
\end{equation*}
gives, derivating and using equation $(2.51)$:
\begin{equation*}
\frac{\dot{u}^0}{u^0}=-K_{ij}\frac{u^i}{u^0}\frac{u^j}{u^0}+
g_{ij}\rho \dot{u}^i \frac{1}{\rho u^0}.\frac{u^j}{u^0}.
\end{equation*}

which shows, since $\frac{u^i}{ u^0}$, $K$, $\rho u^0$, $g$,
$\rho\dot{u}^i$ are bounded, that there exists a constant $A>0$ such
that:
\begin{equation*}
\left|\frac{\dot{u}^0}{u^0}\right|< A.
\end{equation*}
Integrating on $[0,t]$, $0\leq t\leq T^\star$ we obtain:
\begin{equation*}
1\leq u^0(t)\leq u^0(0)\exp(AT^\star).
\end{equation*}
Hence $u^0$ is bounded on $[0,T^\star[$. Now write:
$\frac{1}{\rho}=\frac{1}{\rho u^0}\times u^0$  and
$|u^i|=\frac{|u^i|}{u^0}\times u^0$ to conclude that
$\frac{1}{\rho}$ and $u^i$ are bounded on $[0, T^\star[$.

$\bullet$ \underline{Boundedness of  $\Theta^{0\alpha}$ on $[0, T^\star[$}\\
Since $|H|$, $\rho$, $u^0$, $|K|$, $(\det g)^{-1}$ and $u^i$ are
bounded on $[0,T^\star[$, by equations $(2.56)$ and $(2.57)$ in
$\Theta^{0\alpha}$, there exists two constant $C_1$ and $C_2$: such
that:

\begin{equation*}
|\dot{\Theta}^{0\beta}|\leq
C_1\sum\limits_{\alpha=0}^{3}|\Theta^{0\alpha}| + C_2.
\end{equation*}
Summing in $\alpha$ and integrating on $[0,t]$, $0\leq t< T^\star$,
we obtain:
\begin{equation*}
\sum\limits_{\alpha=0}^3|\Theta^{0\alpha}|\leq
\sum\limits_{\alpha=0}^3|\Theta^{0\alpha}(0)|+4C_2T^\star
+4C_1\int_0^t\sum\limits_{\alpha=0}^3|\Theta^{0\alpha}|ds;
\end{equation*}
and Gronwall Lemma gives:
\begin{equation*}
\sum\limits_{\alpha=0}^3|\Theta^{0\alpha}|(t)\leq
\left(\sum\limits_{\alpha=0}^3|\Theta^{0\alpha}(0)|
+4C_2T^*\right)\exp{(4C_1T^*)}\,\, on \, [0,T^\star[.
\end{equation*}
which shows that each $\Theta^{0\alpha}$ is bounded on  $[0,T^\star[.$\\
$\bullet$ \underline{Boundedness of $e$ on  $[0,T^\star[$}\\
 From the constraint $(2.64)$ we have:
\begin{equation*}
e=-\frac{C^i_{ik}F^{0k}}{u^0};
\end{equation*}
which is bounded on $[0,T^\star[$, since $\frac{1}{u^0}$ and $E$
are. This completes the proof of \\Proposition 5.1.
\end{proof}

\begin{theorem}
Suppose that $\Lambda>0$  and that the mean curvature satisfies
$H(0)<0$. Then the initial values problem for the coupled system of
Einstein-Maxwell-pseudo tensor of pressure has a unique global
solution.
\end{theorem}

 \begin{proof}
 From Proposition 5.1, we have to prove that the mean curvature $H$
is bounded on every interval $[0,T^\star[$ such that
$T^\star<+\infty$. Consider the traceless tensor:
\begin{equation}
\sigma_{ij}=K_{ij}-\frac{1}{3}Hg_{ij}.
\end{equation}
By a direct calculation, we have:
\begin{equation}
\sigma_{ij}\sigma^{ij}=K_{ij}K^{ij}-\frac{1}{3}H^2.
\end{equation}
Using the hamiltonian constraint $(2.61)$ and $(5.9)$, we obtain:
\begin{equation}
\frac{2}{3}H^2-2\Lambda=
\sigma^{ij}\sigma_{ij}+16\pi(\tau_{00}+T_{00})-R
\end{equation}

We have  $\sigma_{ij}\sigma^{ij}\geq0$,  $\tau_{00}\geq0$,
$T_{00}\geq0$. R.T. JANTZEN in [3] and R.WALD in [8] prove that, in
the spacetimes considered here, we always have $R\leq0$. We then
deduce from $(5.10)$:
\begin{equation*}
\frac{2}{3}H^2-2\Lambda>0,
\end{equation*}
 this means:
\begin{equation*}
H^2>3\Lambda.
\end{equation*}
From there we have:
\begin{equation*}
H<-\sqrt{3\Lambda},\,or\,\, H>\sqrt{3\Lambda}
\end{equation*}
 but $H$ is continuous and we have $H(0)<0$, then we must have:
 \begin{equation*}
H<-\sqrt{3\Lambda}.
\end{equation*}
Now if we take $(5.4)$ in which we have $K_{ij}K^{ij}\geq0$, we
conclude that:
\begin{equation*}
H(t)\geq H(0)-\Lambda t.
\end{equation*}
Hence, on every interval $[0,T^\star[$, where $T^\star<+\infty$, we
have:
\begin{equation*}
 H(0)-\Lambda T^\star\leq H(t)\leq-\sqrt{3\Lambda}
\end{equation*}
which prove that $H$ is bounded on each interval $[0,T^\star[$. This
ends the proof of theorem 5.1.
\end{proof}
We end the paragraph by the following result:
\begin{proposition}
For $\Lambda<0$, there exists no global solution to the
Einstein-Maxwell-pseudo-tensor of pressure initial values problem.
\end{proposition}

\begin{proof}
Let $\Lambda<0$ be given and suppose that the system has a global
solution on the whole interval $[0,+\infty[$.\\
\indent By  $(5.3)$ and $(5.9)$, we have:
\begin{equation*}
\frac{dH}{dt}\geq \sigma_{ij}\sigma^{ij}+\frac{1}{3}H^2-\Lambda.\,\,
\end{equation*}

\begin{equation*}
But \,\,\sigma_{ij}\sigma^{ij}\geq0, \,\, then:
\end{equation*}

\begin{equation}
\frac{dH}{dt}\geq\frac{1}{3}H^2-\Lambda.
\end{equation}
Since $\Lambda<0$, we have $-\Lambda>0$ and from $(5.11)$ we deduce
the inequalities:
\begin{equation}
\frac{dH}{dt}\geq\frac{1}{3}H^2,
\end{equation}
\begin{equation}
\frac{dH}{dt}\geq-\Lambda.
\end{equation}
$(5.13)$ gives, by integration on $[0,t]$, $t\geq0$:
\begin{equation}
H(t)\geq H(0)-\Lambda t.
\end{equation}
 $(5.14)$ shows that $H(t)\longrightarrow +\infty$ when $t\longrightarrow +\infty$. So there
 exists $t_0>0$ such that: $H(t_0)>0$. Now we have, given $(5.12)$:
$H(t)\geq W(t)$, for $t\geq t_0$, where $W$ is every function
satisfying:
\begin{numcases}
\strut
 \frac{dW}{dt}=\frac{1}{3}W^2, \label{}\\
W(t_0)=H(t_0)>0.\label{}
\end{numcases}

A solution $W$ of  $(5.15)$ on $[t_0,+\infty[$ is, given $(5.16)$:
\begin{equation}
W(t)=\frac{3H(t_0)}{3-H(t_0)(t-t_0)}.
\end{equation}
 (5.17) shows that: $W(t)\longrightarrow +\infty$ when
$t\longrightarrow_{<} t^\star=t_0+\frac{3}{H(t_0)}>t_0$. Hence
$H(t)\longrightarrow +\infty$ when $t\longrightarrow_{<} t^\star$.
This is impossible since the continuous function $H$ is bounded on
the compact set $[t_0,t^\star]$ and cannot tend to $+\infty$ when
$t$ tends to $t^\star$. So there can exist no global solution if
$\Lambda<0$.
\end{proof}

\section{Asymptotic behavior and geodesic completeness}
Here we extend the results of $[4]$ to the charged case. We also
study the asymptotic behavior of the charge $e$, the matter density
$\rho$,the curvature tensor, the electromagnetic field and the
pseudo-tensor of pressure.

\subsection{Asymptotic behavior}

\begin{proposition}
We have at late times:

\begin{equation}
H(t)=-\sqrt{3\Lambda}+\mathcal{O}(e^{-2\gamma t})
\end{equation}
\begin{equation}
 \dot{H} =\mathcal{O}(e^{-2\mu t})
\end{equation}
\begin{equation}
(\det g)^{-1}=\mathcal{O}(e^{-6\gamma t})
\end{equation}
\begin{equation}
\det g=\mathcal{O}(e^{6\gamma t})
\end{equation}
\begin{equation}
\sigma_{ij}\sigma^{ij}=\mathcal{O}(e^{-2\gamma t})
\end{equation}
\begin{equation}
\sigma_{ij}=\mathcal{O}(e^{\gamma t})
\end{equation}
\begin{equation}
\tau_{00}(t)=\mathcal{O}(e^{-2\gamma t})
\end{equation}
\begin{equation}
T_{00}(t)=\mathcal{O}(e^{-2\gamma t})
\end{equation}
\begin{equation}
R=\mathcal{O}(e^{-2\gamma t})
\end{equation}

\begin{equation}
K_{ij}K^{ij}=\Lambda+\mathcal{O}(e^{-2\gamma t})
\end{equation}

\begin{equation}
g_{ij}(t)= e^{2\gamma t}(G_{ij}+\mathcal{O}(e^{-\gamma t}))
\end{equation}
\begin{equation}
g^{ij}(t)=e^{-2\gamma t}( G^{ij}+\mathcal{O}(e^{-\gamma t}))
\end{equation}
\begin{equation}
\sigma^{ij}(t)=\mathcal{O}(e^{-3\gamma t})
\end{equation}
\begin{equation}
E_{i}E^{i}(t)=\mathcal{O}(e^{-2\gamma t})
\end{equation}
\begin{equation}
F_{ij}F^{ij}(t)=\mathcal{O}(e^{-2\gamma t})
\end{equation}
\begin{equation}
E^{i}(t)=\mathcal{O}(e^{2\gamma t})
\end{equation}
\begin{equation}
F_{ij}=\mathcal{O}(e^{4\gamma t})
\end{equation}
\begin{equation}
F^{ij}\,\, is\,\, bounded
\end{equation}
\begin{equation}
K_{ij}(t)=\mathcal{O}(e^{2\gamma t})
\end{equation}
\begin{equation}
\rho(u^0)^2=\mathcal{O}(e^{-2\mu t})
\end{equation}
\begin{equation}
\rho=\mathcal{O}(e^{-2\mu t})
\end{equation}
\begin{equation}
e=\mathcal{O}(e^{2\gamma t})
\end{equation}
\begin{equation}
\Theta_{ij}(t)=\mathcal{O}(e^{2(\gamma-\mu)t})
\end{equation}
\begin{equation}
\tau_{ij}(t)=\mathcal{O}(e^{8\gamma t})
\end{equation}
\begin{equation}
\gamma^k_{ij}\,\,is\,\,bounded
\end{equation}
\begin{equation}
\Theta^{00}=\Theta_{00}=\mathcal{O}(e^{-2\gamma t})
\end{equation}
\begin{equation}
\tau_{0j}(t)=\mathcal{O}(e^{6\gamma t})
\end{equation}
\begin{equation}
T_{0j}=\mathcal{O}(e^{6\gamma t})
\end{equation}
\begin{equation}
\Theta^{0j}=\mathcal{O}(e^{4\gamma t})
\end{equation}
\begin{equation}
T_{ij}(t)=\mathcal{O}(e^{(6\gamma-2\mu)t})
\end{equation}
where $\gamma^2=\frac{\Lambda}{3}$, $0<\mu<\gamma$, $(G_{ij})$ and
$(G^{ij})$ are positive definite constant matrices.
\end{proposition}

$\underline{proof \,\,of \,\,(6.1)}$\\
$\indent$ We have by $(5.3)$ and $(5.9)$, since
$\sigma_{ij}\sigma^{ij}\geq0$:
\begin{equation}
\frac{dH}{dt}\geq\frac{1}{3}H^2-\Lambda>0.
\end{equation}
But:
$\frac{1}{3}H^2-\Lambda=\frac{1}{3}(-H-\sqrt{3\Lambda})(-H+\sqrt{3\Lambda}).$
We are in the case of global existence, then: $H<-\sqrt{3\Lambda}.$
So: $-H+\sqrt{3\Lambda}>2\sqrt{3\Lambda}.$ We then deduce from
$(6.31)$ that:
\begin{equation}
\frac{dH}{dt}\geq\frac{1}{3}H^2-\Lambda\geq\frac{2
\sqrt{3\Lambda}}{3}(-H-\sqrt{3\Lambda}).
\end{equation}
We can write $(6.32)$ as:
\begin{equation}
\frac{d(H+ \sqrt{3\Lambda})}{dt}+2\gamma(H+\sqrt{3\Lambda})\geq0;
\end{equation}
where
\begin{equation}
 \gamma^2=\frac{\Lambda}{3}.
\end{equation}
  Multiplying $(6.33)$  by $e^{2\gamma t}$ and integrate over $[0,t]$;
  $t>0$,
to obtain:
\begin{equation*}
e^{2\gamma t}(H+ \sqrt{3\Lambda})\geq H(0)+ \sqrt{3\Lambda}
\end{equation*}
which gives:
\begin{equation}
|H+ \sqrt{3\Lambda}|\leq|H(0)+\sqrt{3\Lambda}|e^{-2\gamma t}
\end{equation}
This proves $(6.1)$ with $\gamma$ given by $(6.34)$.   $\Box$\\
$\underline{Proof\,\,of\,\,(6.2)}$\\
We have, for $\mu\in ]0,\gamma[$:
\begin{equation}
e^{2\mu t}\frac{dH}{dt}=e^{2\mu t}\frac{d(H+
\sqrt{3\Lambda})}{dt}=\frac{d}{dt}[e^{2\mu
t}(H+\sqrt{3\Lambda})]-2\mu e^{2\mu t}(H+\sqrt{3\Lambda})
\end{equation}
which proves, given $(6.1)$ and $(6.35)$ since $0<\mu<\gamma$, that
$e^{2\mu t}\frac{dH}{dt}$ is integrable on $[0,+\infty[$. We
conclude that $e^{2\mu t}\frac{dH}{dt}\longrightarrow0$ as $t$ tends
to $+\infty$ and we obtain $(6.2)$. $\Box$\\
$\underline{Proof\,\,of\,\,(6.3)\,\,and\,\,(6.4)}$\\
We have seen that:
\begin{equation}
\frac{d}{dt}[\ln\left(\det g\right)]=-2H
\end{equation}
From $H<-\sqrt{3\Lambda}$ and $(6.36)$ we have:
\begin{equation*}
6\gamma\leq-2H\leq6\gamma+2|H(0)+ \sqrt{3\Lambda}|e^{-2\gamma t}.
\end{equation*}
this means, given $(6.37)$ that:
\begin{equation}
6\gamma\leq \frac{d}{dt}[\ln\left(\det g\right)] \leq6\gamma+2|H(0)+
\sqrt{3\Lambda}|e^{-2\gamma t}.
\end{equation}
Integrating over $[0,t]$, $t>0$,we obtain:
\begin{equation*}
(\det g^0)e^{6\gamma t}\leq \det g\leq(\det
g^0)\exp\left(\frac{1}{\gamma}|H(0)+
\sqrt{3\Lambda}|\right)e^{6\gamma t}.
\end{equation*}
which proves $(6.3)$ and $(6.4)$.\\

$\underline{proof \,\,of \,\,(6.5)}$\\
Since $\tau_{00}\geq0$, $T_{00}\geq0$, and  $-R\geq0$, $(5.10)$
gives:
\begin{equation}
 \frac{2}{3}H^2-2\Lambda\geq\sigma_{ij}\sigma^{ij}\geq0.
\end{equation}
But we have, since $H$ is increasing:
$0\leq-H+\sqrt{3\Lambda}\leq-H(0)+\sqrt{3\Lambda}.$ Then:
\begin{equation*}
 0\leq \frac{2}{3}H^2-2\Lambda\leq2|-H(0)+\sqrt{3\Lambda}||H +\sqrt{3\Lambda}|
\end{equation*}
$(6.5) $ follows then from $(6.39) $ and $(6.1)$.\\
$\underline{proof \,\,of \,\,(6.6)}$\\
$(6.6) $ is proven  as in $[4]$, using $(6.5)$ and the notion of
relative norms.\\
$\underline{proof \,\,of \,\,(6.7), \,\,(6.8)\,\, and \,\,(6.9)}$\\
$(5.10)$ gives, using $(6.5)$:
\begin{equation}
16\pi(\tau_{00}+T_{00})-R+\sigma_{ij}\sigma^{ij}=\frac{2}{3}H^2-2\Lambda=\mathcal{O}(e^{-2\gamma
t})
\end{equation}
But, $T_{00}\geq0$, $\tau_{00}\geq0$, $R\leq0$ and
$\sigma_{ij}\sigma^{ij}\geq0$. Hence $(6.7)$, $(6.8)$ and $(6.9)$
follow from $(6.40)$.\\
$\underline{proof \,\,of \,\,(6.10)}$\\
$(6.10)$ is a direct consequence of $(5.9)$, $(6.5)$ and $(6.1)$.\\
$\underline{proof \,\,of \,\,(6.11)\,\, and \,\,(6.12)}$\\
The proof of $(6.11)$ and $(6.12)$ are given by $[4]$ using
$(6.6)$.\\
$\underline{proof \,\,of \,\,(6.13)}$\\
We have: $\sigma^{ij}=g^{ik}g^{jl}\sigma_{kl}$. $(6.13)$ then
follows from $(6.12)$ and $(6.6)$.\\
$\underline{proof \,\,of \,\,(6.14) \,\, and \,\,(6.15)}$\\
We have:
\begin{equation*}
\tau_{00}=g^{ij}\tau_{ij}=\frac{1}{4}F^{ij}F_{ij}+\frac{1}{2}g_{ij}E^iE^j\geq0
\end{equation*}
Hence,
 \begin{equation*}
0\leq\frac{1}{4}F^{ij}F_{ij}\leq\tau_{00}\,\,
,\,\,0\leq\frac{1}{2}g_{ij}E^iE^j\leq\tau_{00}.
\end{equation*}
$(6.14)$ and $(6.15)$ are then consequences of $(6.7)$.\\
$\underline{proof \,\,of \,\,(6.16)}$\\
The second relation $(3.9)$ gives, using $(6.11)$ and $(6.14)$:
\begin{equation*}
|E^i|\leq \left(g_{rs}E^r E^s\right)^\frac{1}{2}
\left|g\right|^\frac{3}{2}\leq \left( E_r
E^r\right)^\frac{1}{2}\left|g\right|^\frac{3}{2} \leq ce^{-\gamma
t}e^{3\gamma t}=ce^{2\gamma t}
\end{equation*}
which proves $(6.16)$.\\
$\underline{proof \,\,of \,\,(6.17)}$\\
The integration of equation $(2.53)$ gives:
\begin{equation*}
F_{ij}(t)=F_{ij}(0)+C^l_{ij}\int_0^tg_{lm}E^m(s)ds
\end{equation*}
$(6.17)$ then follows from $(6.11)$ and $(6.16)$\\
$\underline{proof \,\,of \,\,(6.18)}$\\
Use $F^{ij}=g^{ik}g^{jl}F_{kl}$, $(6.12)$ and $(6.17)$.\\
$\underline{proof \,\,of \,\,(6.19)}$\\
The relation $(5.8)$ gives:
\begin{equation*}
|K_{ij}|\leq |\sigma_{ij}|+\frac{1}{3}|Hg_{ij}|.
\end{equation*}
$(6.19)$ then follows from $(6.6)$, $(6.1)$ and $(6.11)$.\\
$\underline{proof \,\,of \,\,(6.20)\,\, and \,\,(6.21)}$\\
The relation $(5.1)$ gives, since $\tau_{00}=g^{ij}\tau_{ij}$:
\begin{equation*}
4\pi e^{2\mu t}g^{ij}T_{ij}=e^{2\mu
t}\left[\frac{dH}{dt}+(3\Lambda-H^2)-R+12\pi
T_{00}+8\pi\tau_{00}\right].
\end{equation*}
The relations $(6.1)$, $(6.2)$, $(6.7)$, $(6.8)$, $(6.9)$ and
$0<\mu<\gamma$, give:
\begin{equation*}
g^{ij}T_{ij}=\mathcal{O}(e^{-2\mu t}).
\end{equation*}
But we know that
\begin{equation*}
T_{ij}=\frac{4}{3}\rho u_iu_j+\frac{1}{3}\rho g_{ij},
\end{equation*}
thus:
\begin{equation}
g^{ij}T_{ij}=\frac{4\rho}{3}g^{ij}u_iu_j+\rho\geq
\rho(g^{ij}u_iu_j+1)=\rho(u^0)^2\geq\rho>0.
\end{equation}
$(6.41)$ gives $(6.20)$ and $(6.21)$.\\
$\underline{proof \,\,of \,\,(6.22)}$\\
The constraint $ (2.64)$ $i.e$: $C_{ik}^iE^k+eu^0=0$, $(6.16)$ and
$u^0\geq1$ give $(6.22)$.\\
$\underline{proof \,\,of \,\,(6.23)}$\\
$(6.23)$ is given by $(2.14)$ $i.e$:
$\Theta_{ij}=\frac{\rho}{3}g_{ij}$, $(6.11)$ and $(6.21)$.\\
$\underline{proof \,\,of \,\,(6.24)}$\\
From $(2.48)$ we have:
\begin{equation*}
\tau_{ij}=\frac{1}{2}g_{ij}E^kE_k-g_{ik}g_{jl}E^kE^l
-\frac{1}{4}g_{ij}F^{lk}F_{lk} +g_{jl}F_{ik}F^{lk}
\end{equation*}
this proves $(6.24)$ using $(6.11)$, $(6.14)$, $(6.16)$, $(6.15)$,
$(6.17)$ and $(6.18)$.\\
$\underline{proof \,\,of \,\,(6.25)}$\\
It is a consequence of $(2.19)$, using $(6.11)$ and $(6.12)$.\\
$\underline{proof \,\,of \,\,(6.26)}$\\
From $(2.12)$ we obtain:
\begin{equation*}
\Theta_{00}=T_{00}-\frac{4}{3}\rho(u^0)^2.
\end{equation*}
$(6.26)$ is then a consequence of $(6.8)$ and $(6.20)$ since
$0<\mu<\gamma$.\\
$\underline{proof \,\,of \,\,(6.27)}$\\
From $(2.48)$ we have:
\begin{equation*}
\tau_{0j}=-E^{k}F_{jk}
\end{equation*}
$(6.27)$ is then a consequence of $(6.16)$ and $(6.17)$.\\
$\underline{proof \,\,of \,\,(6.28)}$\\
We have, using the constraint $(2.62)$:
\begin{equation*}
4\pi
T_{0j}=-8\pi\tau_{0j}+g^{ik}\gamma^l_{ki}K_{lj}-g^{ik}\gamma^l_{kj}K_{il}
\end{equation*}
$(6.28)$ is then a consequence of $(6.27)$, $(6.19)$, $(6.24)$ and
$(6.12)$.\\
$\underline{proof \,\,of \,\,(6.29)}$\\
From $(6.12)$ we obtain:
\begin{equation*}
\Theta^{0j}=T^{0j}-\frac{4}{3}\rho
u^0u^j=-g^{ij}T_{0i}-\frac{4}{3}\rho u^0u^j.
\end{equation*}
$(6.29)$ is then a consequence of $(6.28)$, $(6.12)$, $(6.21)$ and
$(3.9)$.\\
$\underline{proof \,\,of \,\,(6.30)}$\\
We have:
\begin{equation*}
T_{ij}=\frac{4}{3}\rho u_iu_j+\frac{1}{3}\rho g_{ij}.
\end{equation*}
 From $(3.9)$ we have:
 \begin{equation}
|T_{ij}|\leq c\rho(u^0)^3|g|^3+\frac{1}{3}|\rho g_{ij}|.
\end{equation}
$(6.30)$ then follows from $(6.12)$, $(6.20)$, $(6.21)$, and
$(6.42)$\\
This ends the proof of proposition 6.1.

\subsection{Geodesic completeness}
The geodesic equations for the metric $(2.1)$ imply that along
geodesics, the variables $t$, $u^0$, $u^i$ satisfy a differential
system which contains, between others, the equation:
\begin{equation}
\frac{dt}{ds}=u^0,
\end{equation}
where $s$ is an affine parameter. The space-time will be
geodesically complete if the affine parameter $s$ tends to $+\infty$
as the time $t$ tends to $+\infty$. Since $(u^0)^2=1+g_{ij}u^iu^j$,
it will be enough if we prove that:
\begin{equation}
\frac{ds}{dt}=(1+g_{ij}u^iu^j)^{-\frac{1}{2}}\geq c>0.
\end{equation}
since by integration we will have: $s\geq c\,t+s_o$, which proves
that
$s$ tends to $+\infty$ as  $t$ tends to $+\infty$.\\
We begin by the following result:
\begin{proposition}
The quantity $g^{ij}u_iu_j$ satisfies the equation:
\begin{equation}
\frac{d}{dt}\left(g^{ij}u_iu_j\right)=\frac{2}{3}Hg^{ij}u_iu_j+2\sigma^{ij}u_iu_j
-\frac{3e_0}{2\rho_0}\exp\left(\frac{3}{4}\int_0^t\frac{1}{u^0(s)}ds\right)E^iu_i.
\end{equation}
\end{proposition}
\begin{proof}
From the evolution equation $(2.51)$ and since $g_{ij}$ is
symmetric, we have:
\begin{equation*}
\frac{d}{dt}(g^{ij}u_iu_j)=\frac{d}{dt}(g_{ij}u^iu^j)=-2K_{ij}u^iu^j+2g_{ij}\dot{u}^iu^j.
\end{equation*}
We can write this equation as:
\begin{equation}
\frac{d}{dt}(g^{ij}u_iu_j)=-2K^{ij}u_iu_j+2\dot{u}^iu_i.
\end{equation}
We have, using equation $(2.55)$ of $u^i$:
\begin{equation}
2\dot{u}^iu_i=4K^{ij}u_iu_j-\frac{2\gamma_{jk}^iu^ju^ku_i}{u^0}+\frac{3}{2}\frac{1}{\rho
u^0}C_{jk}^jE^kE^iu_i-\frac{3}{2}\frac{1}{\rho(
u^0)^2}C_{jk}^jE^kF_{ml}u^mu^l\,\,
\end{equation}
The last term is zero since $F_{ml}u^mu^l=0,$ and the constraint
$(2.64)$ gives:   $C^j_{jk}E^k=-eu^0$. Equation $(6.47)$ then
writes:
\begin{equation}
2\dot{u}^iu_i=4K^{ij}u_iu_j-\frac{3}{2}\frac{e}{\rho}E^iu_i-\frac{2\gamma_{jk}^iu^ju^ku_i}{u^0}\,\,.
\end{equation}
We have from $(6.46)$, $(6.47)$ and $(6.48)$:
\begin{equation}
\frac{d}{dt}(g^{ij}u_iu_j)=2K^{ij}u_iu_j-\frac{3}{2}\frac{e}{\rho}E^iu_i-\frac{2\gamma_{ij}^lu^iu^ju_{_l}}{u^0}\,\,.
\end{equation}
But we can write from expression $(2.19)$ and $\gamma^l_{ij}$:
\begin{eqnarray*}
2\gamma_{ij}^lu^iu^ju_{_l}&=&-g_{im}(C_{jk}^mu^ju^k)u^i\\
& &+g_{jm}(C_{ki}^mu^ku^i)u^j\\
& &+g_{km}(C_{ij}^mu^iu^j)u^k=0
\end{eqnarray*}
since $C_{ij}^m=-C_{ji}^m$. Then, use the traceless tensor
$\sigma_{ij}=K_{ij}-\frac{H}{3}g_{ij}$, $(6.49)$ gives:
\begin{equation}
\frac{d}{dt}(g^{ij}u_iu_j)=2\sigma^{ij}u_iu_j+\frac{2}{3}Hg^{ij}u_iu_j-\frac{3}{2}
\frac{e}{\rho}E^iu_i.
\end{equation}
Now we know by $(6.26)$ that $\rho$ satisfies:
\begin{equation*}
\dot{\rho}=-(\frac{^4\nabla_\alpha\,\,^4u^\alpha}{^4u^0}+\frac{3}{4}\frac{1}{^4u^0})\rho
\end{equation*}
Integrating   over $[0,t]$, $t>0$, we obtain:
\begin{equation*}
\rho(t)=\rho(0)\exp\left(-\int_0^t\left(\frac{^4\nabla_\alpha\,\,^4u^\alpha}{^4u^0}+
\frac{3}{4}\frac{1}{^4u^0}\right)ds\right).
\end{equation*}
By $(2.39)$ we have:
\begin{equation*}
e(t) =
e(0)\exp\left(-\int_0^t\frac{\,^4\nabla_\alpha\,\,^4u^\alpha}{^4u^0}ds\right).
\end{equation*}
Hence:
\begin{equation}
\frac{e(t)}{\rho(t)}=\frac{e_0}{\rho_0}\exp(\frac{3}{4}\int_0^t\frac{1}{^4u^0}ds).
\end{equation}
$(6.45)$ is then a direct consequence of $(6.50)$ and $(6.51)$.
\end{proof}

\begin{proposition}
If $\sqrt{\frac{\Lambda}{3}} > \frac{3}{4}$\,\,, the space-time is
geodesically complete.
\end{proposition}
\begin{proof}
Since by $(6.13)$ $\sigma^{ij}=\mathcal{O}(e^{-3\gamma t})$,
$e^{3\gamma t}\sigma^{ij}$ is bounded. The matrix $G^{ij}$ in
proposition 6.1 being constant and positive definite, we have:
\begin{equation}
e^{3\gamma t}\sigma^{ij}u_iu_j\leq CG^{ij}u_iu_j
\end{equation}
where $C$ is a constant. By $(6.12)$ we can write:
\begin{equation}
G^{ij}u_iu_j\leq Ce^{2\gamma t}g^{ij}u_iu_j.
\end{equation}
Since $g$ is a scalar product,we have:
\begin{equation}
-E^iu_i\leq
\left(E^iE_i\right)^{\frac{1}{2}}\left(g^{ij}u_iu_j\right)^{\frac{1}{2}}.
\end{equation}
So, $(6.1)$, $(6.45)$, $(6.52)$, $(6.53)$, $(6.54)$ and $(6.14)$
imply:
\begin{equation}
\frac{d}{dt}\left(g^{ij}u_iu_j\right)\leq (-2\gamma+ce^{-2\gamma
t})g^{ij}u_iu_j+ce^{-\gamma t}g^{ij}u_iu_j+ c\exp\left(-\gamma
t+\frac{3}{4}\int_0^t\frac{1}{u^0(s)}ds\right)\left(g^{ij}u_iu_j\right)^{\frac{1}{2}}.
\end{equation}
From where, we deduce:
\begin{equation}
\frac{d}{dt}\left(g^{ij}u_iu_j\right)\leq (-2\gamma+ce^{-\gamma
t})g^{ij}u_iu_j+ c\exp\left(-\gamma
t+\frac{3}{4}\int_0^t\frac{1}{u^0(s)}ds\right)\left(g^{ij}u_iu_j\right)^{\frac{1}{2}}.
\end{equation}
Sine $u^0\geq1$, we have $\int^t_0\frac{ds}{u^0}\leq t$ and $-\gamma
t +\frac{3}{4}\int^t_0\frac{ds}{u^0}\leq
t\left(-\gamma+\frac{3}{4}\right)$\\
So if: $-\gamma+\frac{3}{4}<0$ or, if we set:
$\omega=\gamma-\frac{3}{4}>0$, we deduce from $(6.56)$:
\begin{equation}
\frac{d}{dt}\left(g^{ij}u_iu_j\right)\leq (-2\omega+ce^{-\omega
t})g^{ij}u_iu_j+ ce^{-\omega
t}\left(g^{ij}u_iu_j\right)^{\frac{1}{2}}
\end{equation}
Let us set:
\begin{equation*}
Z=e^{\omega t}g^{ij}u_iu_j.
\end{equation*}
We have
\begin{equation*}
\frac{dZ}{dt}=\omega Z+ e^{\omega
t}\frac{d}{dt}\left(g^{ij}u_iu_j\right),
\end{equation*}
Hence, from $(6.57)$ we obtain:
\begin{equation*}
\frac{dZ}{dt}\leq \omega Z+(-2\omega+ce^{- \omega t}
)Z+ce^{-\frac{\omega}{2}t}Z^{\frac{1}{2}}.
\end{equation*}
from where we have:
\begin{equation}
\frac{dZ}{dt}\leq ce^{-\frac{\omega}{2}t}Z +
ce^{-\frac{\omega}{2}t}Z^{\frac{1}{2}}.
\end{equation}
But
\begin{equation*}
e^{-\frac{\omega}{2}t}Z^{\frac{1}{2}}=e^{-\frac{\omega}{4}t}\left(e^{-\frac{\omega}{2}t}Z\right)^{\frac{1}{2}}
\leq
\frac{1}{2}\left[e^{-\frac{\omega}{2}t}+e^{-\frac{\omega}{2}t}Z\right].
\end{equation*}
We then deduced, from $(6.58)$ that:
\begin{equation}
\frac{d(Z+1)}{dt}\leq ce^{-\frac{\omega}{2}t}(Z +1).
\end{equation}
$(6.59)$ proves that $Z=e^{\omega t}(g^{ij}u_iu_j)$ is bounded, and
then, $g^{ij}u_iu_j =g_{ij}u^iu^j$ is bounded. So there exists $A>1$
such that
\begin{equation*}
1\leq u^0\leq A.
\end{equation*}
We then have $(6.44)$ and this ends the proof of proposition 6.3.
\end{proof}

\section{The positivity conditions}
Let $(\,\,^4X^\alpha)$ be a future pointing vector. The quantity
\begin{equation*}
(\,\,^4\tau_{\alpha\beta}+\,\,
^4T_{\alpha\beta})\,\,^4X^{\alpha}\,\,^4X^\beta
\end{equation*}
represents physically, the density of energy of a charged particle,
measured by an observant whose speed is $(\,\,^4X^\alpha)$. Hence,
this quantity should always be positive. Recall that a physical
theory always has to satisfy at least one positivity condition
$[2]$. There are three types of positivity conditions, they are, for
$(\,\,^4X^\alpha)$, $(\,\,^4Y^\alpha)$  future pointing vectors:

$a)$ The weak positivity condition which means:
\begin{equation}
(\,\,^4\tau_{\alpha\beta}+\,\,
^4T_{\alpha\beta})\,\,^4X^{\alpha}\,\,^4X^\beta\geq0
\end{equation}
$b)$ The strong positivity condition which means:
\begin{equation}
 ^4R_{\alpha\beta}\,\,^4X^{\alpha}\,\,^4X^\beta\geq0;
\end{equation}
$c)$ The dominant energy condition which means:
\begin{equation}
(\,\,^4\tau_{\alpha\beta}+\,\,
^4T_{\alpha\beta})\,\,^4X^{\alpha}\,\,^4Y^\beta\geq0.
\end{equation}
Obviously the dominant energy condition implies the weak energy
condition, just setting\\ $\,\, ^4Y^\alpha=\,\, ^4X^\alpha$. We
begin by proving:
\begin{proposition}
 Let $(\,\,^4X^\alpha)$ and $(\,\,^4Y^\alpha)$ be two future
 pointing vectors. Then:
\begin{equation}
^4X^\alpha \,\, ^4Y_\alpha\leq0.
\end{equation}
\end{proposition}
\begin{proof}
Since $\,\,^4X^\alpha \,\,
^4Y_\alpha=\,\,^4g_{\alpha\beta}\,\,^4X^\alpha\,\,^4Y^\beta$, the
definition $(2.1)$ of the metric $^4g$ implies that $(7.4)$ is
equivalent to:
\begin{equation}
-\,\,^4X^0\,\,^4Y^0+g_{ij}\,\,^4X^i\,\,^4Y^j \leq0.
\end{equation}
But ($\,\,^4X^\alpha)$ and $(\,\,^4Y^\alpha)$ are future pointing,
thus we have:
\begin{equation*}
^4X^\alpha \,\,^4X_\alpha\leq0,\,\,^4X^0\geq0,\,\,^4Y^\alpha
\,\,^4Y_\alpha\leq0,\,\,^4Y^0\geq0
\end{equation*}
or equivalently:
\begin{equation*}
0\leq g_{ij}\,\,^4X^i\,\,^4X^j\leq(\,\,^4X^0)^2;\,\,^4X^0\geq0;
\,\,0\leq g_{ij}\,\,^4Y^i\,\,^4Y^j\leq(\,\,^4Y^0)^2;\,\,^4Y^0\geq0.
\end{equation*}
So, taking the square roots and the products:
\begin{equation}
0\leq(g_{ij}\,\,^4X^i\,\,^4X^j)^\frac{1}{2}(g_{ij}\,\,^4Y^i\,\,^4Y^j)^\frac{1}{2}\leq\,\,^4X^0\,\,^4Y^0.
\end{equation}
Since $g=(g_{ij})$ is a scalar product, we have the Schwartz
inequality:
\begin{equation}
|g_{ij}\,\,^4X^i\,\,^4X^j|\leq(g_{ij}\,\,^4X^i\,\,^4X^j)^\frac{1}{2}(g_{ij}\,\,^4Y^i\,\,^4Y^j)^\frac{1}{2}.
\end{equation}
  Hence, $(7.4)$ is a direct consequence of $(7.6)$ and
$(7.7)$.
\end{proof}
We now prove an important result.
\begin{proposition}
Let $(\,\,^4X^\alpha)$ and $(\,\,^4Y^\alpha)$ be two future pointing
vectors. The Maxwell tensor $^4\tau_{\alpha\beta}$ satisfies:
\begin{equation}
^4\tau_{\alpha\beta}\,\,^4X^\alpha\,\,^4Y^\beta\geq0.
\end{equation}
\end{proposition}
\begin{proof}
It shall be enough if we prove the inequality $(7.8)$ in a
particular frame. Guided by the electromagnetic field itself, we
choose a frame of four future pointing vectors $l=(l^\alpha)$;
$n=(n^\alpha)$; $x=(x^\alpha)$; $y=(y^\alpha)$ such that:
\begin{equation}
l_\alpha l^\alpha=n_\alpha n^\alpha=l_\alpha x^\alpha=l_\alpha
y^\alpha=n_\alpha y^\alpha=n_\alpha x^\alpha=0.
\end{equation}
The antisymmetric 2-form $\,\,^4F_{\alpha\beta}$ can be written in
one of two forms:
\begin{equation}
^4F_{\alpha\beta}=\frac{A}{2}(l_\alpha n_\beta-l_\beta n_\alpha) +
\frac{B}{2}(x_\alpha y_\beta-x_\beta y_\alpha)
\end{equation}
or
\begin{equation}
^4F_{\alpha\beta}=\frac{C}{2}(l_\alpha x_\beta-l_\beta x_\alpha),
\end{equation}
Where $A$, $B$, $C$ are  constants. It is important to choose the
constants $A$, $B$, $C$ such that:
\begin{equation}
l_\alpha n^\alpha=-1;\,\, x_\alpha x^\alpha=y_\alpha y^\alpha=1;\,\,
x_\alpha y^\alpha=0.
\end{equation}
Let us consider the Maxwell tensor:
\begin{equation}
^4\tau_{\alpha\beta}=-\frac{1}{4}\,\,^4g_{\alpha\beta}\,\,^4F^{\lambda\mu}\,\,^4F_{\lambda\mu}
+ \,\,^4F_{\alpha\lambda}\,\,^4F_\beta^{\,\,\lambda}.
\end{equation}
$(7.11)$ gives, using $(7.9)$ and $(7.2)$:
\begin{equation*}
^4F^{\lambda\mu}\,\,^4F_{\lambda\mu}=0;\,\,^4F_{\alpha\lambda}\,\,^4F_\beta^{\,\,\lambda}=\frac{c^2}{4}l_\alpha
l_\beta.
\end{equation*}
$(7.13)$  then gives:
\begin{equation}
^4\tau_{\alpha\beta}=\frac{c^2}{4}l_\alpha l_\beta.
\end{equation}
Let $(\,\,^4X^\alpha)$ and $(\,\,^4Y^\alpha)$ be two future pointing
vectors. We have:
\begin{equation*}
^4\tau_{\alpha\beta}\,\,^4X^\alpha
\,\,^4Y^\beta=\frac{c^2}{4}(l_\alpha \,\,^4X^\alpha)(l_\beta
\,\,^4Y^\beta)
\end{equation*}
But since $(l^\alpha)$, $(\,\,^4X^\alpha)$, $(\,\,^4Y^\alpha)$ are
future pointing vectors, by $(7.4)$ we have
$l_\alpha\,\,^4X^\alpha\leq0$ and $l_\beta\,\,^4Y^\beta\leq0$. Then
$^4\tau_{\alpha\beta}\,\,^4X^\alpha\,\,^4Y^\beta\geq0$ and this
proves proposition 7.2
\end{proof}

\begin{theorem}
The global solution of the Einstein-Maxwell system with
pseudo-tensor of pressure satisfies:

$1^0)$ The weak positivity condition;

$2^0)$ The strong  positivity condition if
$\Lambda\geq\frac{\left(H(0)\right)^2}{2}$.
\end{theorem}
\begin{proof}-\\
$1^0)$ Let $(\,\,^4X^\alpha)$ be a future pointing vector. By
$(7.8)$ we have:
\begin{equation}
^4\tau_{\alpha\beta}\,\,^4X^\alpha \,\,^4X^\beta\geq0.
\end{equation}
Now by definition $(2.12)$ of $\,\,^4T_{\alpha\beta}$, we have:
\begin{equation}
^4T_{\alpha\beta} =\frac{4}{3}\rho \,\,^4u_\alpha \,\,^4u_\beta
+\,\,^4\Theta_{\alpha\beta}
\end{equation}
If $(\,\,^4X^\alpha)$ is a future pointing vector and since
$(\,\,^4u^\alpha)$ is by definition a future pointing vector,
$(7.4)$ implies, since $\rho>0$:
\begin{equation}
\left(\frac{4}{3}\rho \,\,^4u_\alpha
\,\,^4u_\beta\right)\,\,^4X^\alpha \,\,^4X^\beta= \frac{4}{3}\rho
\left(\,\,^4u_\alpha \,\,^4X^\alpha\right)\left(\,\,^4u_\beta
\,\,^4X^\beta\right)\geq0
\end{equation}
Now we use the definition of the pseudo-tensor of pressure
$\,\,^4\Theta_{\alpha\beta}$. We have, using $(2.13)$ and $(2.14)$,
for a future pointing vector:
\begin{eqnarray*}
^4\Theta_{\alpha\beta}\,\,^4X^\alpha \,\,^4X^\beta
&=&\,\,^4\Theta_{00}\left(\,\,^4X^0\right)^2+2\,\,^4\Theta_{0i}\,\,^4X^0\,\,^4X^i
+\,\,^4\Theta_{ij}\,\,^4X^i\,\,^4X^j\\
&=&\left(\,\,^4Z_0\right)^2\left(\,\,^4X^0\right)^2+2\left(\frac{1}{2}
\,\,^4Z_0\,\,^4Z_i\right)\,\,^4X^0\,\,^4X^i+\frac{1}{3}\rho
g_{ij}\,\,^4X^i\,\,^4X^j\\
&=&\left(\,\,^4Z_0\right)^2\left(\,\,^4X^0\right)^2+\left(
\,\,^4Z_0\,\,^4X^0\right)\,\,^4Z_i\,\,^4X^i+ \frac{1}{3}\rho
g_{ij}\,\,^4X^i\,\,^4X^j\\
&=&\left(\,\,^4Z_0\right)^2\left(\,\,^4X^0\right)^2+\left(
\,\,^4Z_0\,\,^4X^0\right)\left(\,\,^4Z_\alpha\,\,^4X^\alpha-\,\,^4Z_0\,\,^4X^0\right)+
\frac{1}{3}\rho g_{ij}\,\,^4X^i\,\,^4X^j
\end{eqnarray*}
\begin{equation}
^4\Theta_{\alpha\beta}\,\,^4X^\alpha
\,\,^4X^\beta=\left(\,\,^4Z_0\,\,^4X^0\right)\left(\,\,^4Z_\alpha\,\,^4X^\alpha\right)+
\frac{1}{3}\rho g_{ij}\,\,^4X^i\,\,^4X^j
\end{equation}
But since $(\,\,^4Z^\alpha)$ is future pointing, we have
$\,\,^4Z^0\geq0$, then $\,\,^4Z_0=-\,\,^4Z^0\leq0$, thus,
$\,\,^4Z_0\,\,^4X^0\leq0$ since $\,\,^4X^0\geq0$.\\
Using once more $(7.4)$ we have $\,\,^4Z_\alpha\,\,^4X^\alpha\leq0$,
and finally:
$\left(\,\,^4Z_0\,\,^4X^0\right)\left(\,\,^4Z_\alpha\,\,^4X^\alpha\right)\geq0$.
Now the metric $(g_{ij})$ is positive definite, thus
$g_{ij}\,\,^4X^i\,\,^4X^j\geq0$. In conclusion $(7.18)$ implies\\
$^4\Theta_{\alpha\beta}\,\,^4X^\alpha \,\,^4X^\beta\geq0$, and using
$(7.17)$ and $(7.18)$ we obtain:
\begin{equation}
^4T_{\alpha\beta}\,\,^4X^\alpha \,\,^4X^\beta\geq0.
\end{equation}
Finally we have by $(7.15)$ and $(7.19)$:
\begin{equation*}
\,\,^4\tau_{\alpha\beta}\,\,^4X^\alpha \,\,^4X^\beta
+\,\,^4T_{\alpha\beta}\,\,^4X^\alpha
\,\,^4X^\beta=\left(\,\,^4\tau_{\alpha\beta} +
\,\,^4T_{\alpha\beta}\right)\,\,^4X^\alpha \,\,^4X^\beta\geq0
\end{equation*}
and the weak positivity condition is proved.\\
$2^0)$ Let $(\,\,^4X^\alpha)$ be a future pointing vector.\\
The Einstein's equations $(2.8)$ imply:
\begin{equation}
^4R_{\alpha\beta}\,\,^4X^\alpha
\,\,^4X^\beta=\left(\,\,\frac{^4R}{2}-\Lambda\right)\,\,^4g_{\alpha\beta}\,\,^4X^\alpha\,\,^4X^\beta+8\pi\left(
\,\,^4\tau_{\alpha\beta}+\,\,^4T_{\alpha\beta}\right)\,\,^4X^\alpha\,\,^4X^\beta.
\end{equation}
By the weak positivity condition we know that:
\begin{equation}
8\pi\left(
\,\,^4\tau_{\alpha\beta}+\,\,^4T_{\alpha\beta}\right)\,\,^4X^\alpha\,\,^4X^\beta\geq0.
\end{equation}
Let us see at what condition we will also have:
\begin{equation*}
\left(\frac{^4R}{2}-\Lambda\right)\,\,^4g_{\alpha\beta}\,\,^4X^\alpha\,\,^4X^\beta=
\left(\frac{^4R}{2}-\Lambda\right) \,\,^4X^\alpha\,\,^4X_\alpha\geq0
\end{equation*}
By definition, we have:
\begin{equation}
^4R=\,\,^4g^{\alpha\beta}\,\,^4R_{\alpha\beta}=\,\,^4g^{00}\,\,^4R_{00}+g^{ij}\,\,^4R_{ij}=
-\,\,^4R_{00}+g^{ij}\,\,^4R_{ij}.
\end{equation}
We know that $\,\,^4R_{ij}$ and $R_{ij}$ are linked by:
\begin{equation}
^4R_{ij}=R_{ij}-\partial_tK_{ij}+HK_{ij}-2K_{il}K_j^{\,\,\,l}.
\end{equation}
Contracting with $g^{ij}$ we obtain:
\begin{equation}
g^{ij}\,\,^4R_{ij}=R-g^{ij}\partial_tK_{ij}+H^2-2K_{ij}K^{ij}.
\end{equation}
Now we have, using equation $(2.51)$:
\begin{equation}
g^{ij}\partial_tK_{ij}=\partial_t(g^{ij}K_{ij})-K_{ij}\partial_tg^{ij}=\partial_tH
-2K_{ij}K^{ij}
\end{equation}
From $(7.24)$ and $(7.25)$ we obtain:
\begin{equation}
g^{ij}\,\,^4R_{ij}=R-\partial_tH+H^2.
\end{equation}
Now in the Einstein's equations $(2.8)$ take $\alpha=\beta=0$ to
obtain:
\begin{equation}
^4R_{00}=-\frac{^4R}{2}+\Lambda+8\pi(\,\,^4\tau_{00}+\,\,^4T_{00}).
\end{equation}
Then using $(7.22)$, $(7.26)$ and $(7.27)$ we obtain:
\begin{equation*}
^4R=\frac{^4R}{2}-\Lambda-8\pi(\,\,^4\tau_{00}+\,\,^4T_{00})+R-\partial_tH+H^2.
\end{equation*}
From where we deduce, since $\,\,^4\tau_{00}+T_{00}\geq0$, $R\leq0$
and $\partial_tH\geq0$ (see (6.31)):
\begin{equation*}
\frac{^4R}{2} \leq-\Lambda+H^2
\end{equation*}
But   $H$ is increasing and negative, then $H^2\leq H^2(0)$; hence:
\begin{equation*}
\frac{^4R}{2}-\Lambda\leq -2\Lambda+ \left(H(0)\right)^2.
\end{equation*}
So we will have $\frac{^4R}{2}-\Lambda\leq0 $ if $-2\Lambda+
\left(H(0)\right)^2\leq0 $ or
$\Lambda\geq\frac{\left(H(0)\right)^2}{2}$. Since $(\,\,^4X^\alpha)$
is future pointing, we have by $(7.4)$
$\,\,^4X^\alpha\,\,^4X_\alpha\leq0.$ In conclusion if
$\Lambda\geq\frac{\left(H(0)\right)^2}{2}$, then:
\begin{equation}
\left(\frac{^4R}{2}-\Lambda\right)\,\,^4X^\alpha\,\,^4X_\alpha\geq0
\end{equation}
$(7.28)$ and $(7.21)$ imply
\begin{equation*}
^4R_{\alpha\beta}\,\,^4X^\alpha \,\,^4X^\beta\geq0.
\end{equation*}
This ends ends the proof of theorem $(7.1)$)
\end{proof}

\section{Well-posedness}
According to Hadamard, a mathematical problem is well-posed if its
solution exists, if the solution is unique, and if the solution is a
continuous function of the initial data.\\
We have only to prove the last point. in this paragraph we suppose
$e_0>0.$ The initial data is denoted $V_0$ where:
\begin{equation*}
V_0=(g^0,K^0,F^0,E^0,U^{0,i},\Theta^{0,\alpha},\rho_{_0},e_{_0})\in
\Omega:=\mathcal{S}_3(\mathbb{R})\times\mathcal{S}_3(\mathbb{R})\times\mathcal{A}_3(\mathbb{R})\times
\mathbb{R}^{10}\times]0,\infty[\times]0,\infty[
\end{equation*}
where $\mathcal{S}_3(\mathbb{R})$ and $\mathcal{A}_3(\mathbb{R})$
are respectively the sets of $3\times3$ symmetric and antisymmetric
matrices. We suppose that the initial data satisfy the constraints
$(2.61)$, $(2.62)$, $(2.63)$, $(2.64)$. The global solution:
\begin{equation*}
S=(g,K,F,E,u^i,\theta^{0\alpha},\rho,e)
\end{equation*}
of the evolution system is also in $\Omega$ and is a function:
$$\begin{array}{lcll}
S&: \Omega\times [0,\infty[&\longrightarrow &\Omega\\
 & (V_0,t)&\longmapsto &      S(V_0,t),
\end{array}$$
We can also denote:
$$\begin{array}{lcll}
S(V_0,.)&:  [0,\infty[&\longrightarrow &\Omega\\
 & t&\longmapsto &      S(V_0,t),
\end{array}$$
such that: $S(V_0,0)=V_0$. Every component of $S$ is of class $C^1$.
Note that $\Omega$ is an open set of $\widetilde{\Omega}
=\mathcal{S}_3(\mathbb{R})\times\mathcal{S}_3(\mathbb{R})\times\mathcal{A}_3(\mathbb{R})\times
\mathbb{R}^{12}$. We can write:
$$\begin{array}{lcll}
S&:  \Omega &\longrightarrow &\mathcal{C}^1([0,\infty[,
\widetilde{\Omega})\\
 & V_0 &\longmapsto &      S(V_0,.),
\end{array}$$
If $\Gamma$ is a compact set of $[0,\infty[$, we define the seminorm
$P_\Gamma$ on $\mathcal{C}^1([0,\infty[,\widetilde{\Omega})$ by:
\begin{equation*}
P_\Gamma(\varphi)=\sup_{t\in\Gamma}\left( |\varphi|_1\right),
\varphi\in \mathcal{C}^1([0,\infty[,\widetilde{\Omega})
\end{equation*}
Where $\left| W\right|_1=\sum\limits_{i=1}^n\left| W_i\right|$ with
$W=(W_1,...,W_n)\in \mathbb{R}^n$. $\mathcal{C}^1([0,\infty[,
\widetilde{\Omega})$ is a locally convex topological vector space,
whose topology $\tau$ is generated by the family of seminorms
$P_\Gamma.$ We will prove that $S$ is continuous from $\Omega$ to
$\mathcal{C}^1([0,\infty[, \widetilde{\Omega})$, endowed with the
topology $\tau$. Let $p\in \mathbb{N}^\star$. Set:
\begin{equation*}
\Omega_p=\mathcal{B}(0,2^p)\times]2^{-p},2^p[\times]2^{-p},2^p[
\end{equation*}
where $\mathcal{B}(0,2^p)$ is the ball of radius $2^p$ of
$\mathcal{S}_3(\mathbb{R})\times\mathcal{S}_3(\mathbb{R})\times\mathcal{A}_3(\mathbb{R})\times
\mathbb{R}^{10}$. Notice that:
\begin{equation*}
\overline{\Omega}_p\subset\Omega_{p+1}\,\,\,\,\, and\,\,\,\,\,\,
\bigcup\limits_{p\in\mathbb{N}}\Omega_p=\Omega
\end{equation*}
The function $S$ will be continuous on  $\Omega$, if its restriction
to every  $\Omega_p$ is continuous. \\
In what follows, we associate to every initial data the iterated
sequence defined in paragraph 4. Recall that  the iterated sequence
converges to a unique solution if $\Lambda>0$ and
$H(0)<0$.\\
For every $n \in\mathbb{N}$, the iterated sequence
\begin{equation*}
V_n=(g_n,K_n,F_n,E_n,u^i_n,\theta_n^{0\alpha},\rho_n,e_n)
\end{equation*}
is a function:
$$\begin{array}{lcll}
V_n&: \Omega\times [0,\infty[&\longrightarrow &\Omega\\
 & (V_0,t)&\longmapsto &      V_n(V_0,t),
\end{array}$$
\begin{proposition}
$V_n$ is continuons on $\Omega\times[0,\infty[$.
\end{proposition}

\begin{proof}
For $ V_0\in\Omega$, $ t\in[0,\infty[$,  we have by definition
$V_0(V_0,t)=V_0$; hence $V_0$ is continuous on
$\Omega\times[0,\infty[$. Suppose that for $ n\in\mathbb{N}$,
$V_0,\,\,V_1,...,\,\,V_n$ are continuous, then the function:
\begin{equation}
V_{n+1}(V_0,t)=\int_0^tf\circ V_n(V_0,s)ds+V_0
\end{equation}
where $f$ defines the right hand side of the evolution system, is
also continuous on $\Omega\times[0,\infty[$. Hence $V_n$ is
continuous for every $ n\in\mathbb{N}$.
\end{proof}

\begin{proposition}
For $p\in\mathbb{N}$, $0<T<+\infty$, the sequence $(V_n)_n$ is
uniformly bounded on  $\Omega_p\times[0,T[$
\end{proposition}
\begin{proof}
Let us recall that:
\begin{equation}
H(0)=g^{0,ij}K^0_{ij}=H_0(V_0).
\end{equation}
$(8.2)$ shows that $H_0$ is continuous on
 $\overline{\Omega}_p\times[0,T]$, hence it is bounded on it. Under
 the hypothesis that the iterated sequence satisfy the constraints
 $(2.61)$, $(2.62)$, $(2.63)$, $(2.64)$, we must have:
 \begin{equation*}
H_0(V_0)\leq H_n(V_0,t)\leq -\sqrt{3\Lambda},\forall
n\in\mathbb{N},\,\,(V_0,t)\in\Omega\times[0,\infty[.
\end{equation*}
$H_0$ is bounded on $\Omega_p\times[0,T[$ and it does not depend on
$t$. So there exists a constant $C_p>0$ depending only on $p$ such
that:
\begin{equation*}
-C_p\leq H_n(V_0,t)\leq  -\sqrt{3\Lambda},\forall
n\in\mathbb{N},\,\,(V_0,t)\in\Omega_p\times[0,T[.
\end{equation*}
This shoes that the sequence $(H_n)_{n\in \mathbb{N}}$ is uniformly
bounded on $\Omega_p\times[0,T[$.
\begin{itemize}
\item[$1^0)$] \underline{boundedness of $(g_n)_{n\in \mathbb{N}}$ on
$\Omega_p\times[0,T[$}\\
The sequence $(V_n)_n$ satisfy the Hamiltonian constraint $(2.61)$,
that is:
\begin{equation}
H_n^2=16\pi(\tau_{n,00}+T_{n,00})+K_{n,ij}K_n^{ij}-R_n.
\end{equation}
From the fact that: $\Lambda\geq0$, $\tau_{n,00}\geq0$,
$T_{n,00}\geq0$, $K_{n,ij}K_n^{ij}\geq0$, $-R_n\geq0$ and $(H_n)_n$
bounded on $\Omega_p\times[0,T[$, $(8.3)$ shows that the sequence
$(K_{n,ij}K_n^{ij})_n$ of positive terms is bounded on
$\Omega_p\times[0,T[$. Using the definition of the iterated
sequence, the integration of equation $(2.61)$ on $[0,t]$, $0\leq
t<T$ gives:
\begin{equation}
|g_{n+1}(V_0,t)|\leq |g^0|+2\int_0^t|K_n(V_0,s)|ds.
\end{equation}
Using the notion of relative norm introduced in paragraph $3.3$, we
deduce  from $(8.4)$ the inequality:
\begin{equation}
|g_{n+1}(V_0,t)|\leq
C\left[|g^0|+\int_0^t||g_n(V_0,s)||.||K_n(V_0,s)||_{g_n(V_0,s)}ds\right]
\end{equation}
where $C>0$ is a constant. From $(3.6)$ and $(3.7)$, $(8.5)$
implies:
\begin{equation}
||g_{n+1}(V_0,t)||\leq
C\left[||g^0||+\int_0^t(K_{n,ij}K_n^{ij})^\frac{1}{2}||g_n(V_0,s)||ds\right].
\end{equation}
But $(K_{n,ij}K_n^{ij})_n$ is uniformly bounded on
$\Omega_p\times[0,T[$. Then $(8.6)$ implies that there exists a
constant $C_p>0$, depending only on $p$, such that:

\begin{equation*}
||g_{n+1}(V_0,t)||\leq
C_p\left[||g^0||+\int_0^t||g_n(V_0,s)||ds\right],
\,\,\,\,\,\forall\,\, (V_0,t)\in\Omega_p\times[0,T[.
\end{equation*}
 By induction on $n$, we obtain:
\begin{equation*}
||g_{n+1}(V_0,t)||\leq
C_p||g^0||\left[1+C_pt+\frac{(C_pt)^2}{2}+...+\frac{(C_pt)^n}{n}+||g^0||\frac{(C_pt)^{n+1}}{n+1}\right],\,\,\forall\,
(V_0,t)\in\Omega_p\times[0,T[.
\end{equation*}
Finally, we obtain:
\begin{equation}
|g_{n+1}(V_0,t)|\leq C_p|g^0|(1+|g^0|)\exp(C_pt),\,\forall
n\in\mathbb{N},\,(V_0,t)\in\Omega_p\times[0,T[.
\end{equation}

$(8.7)$ shows that the sequence $(g_n)_{n\in \mathbb{N}}$ is
uniformly bounded on $\Omega_p\times[0,T[$.

\item[$2^0)$] \underline{Boundedness of the sequence $(K_n)_{n\in \mathbb{N}}$ on
$\Omega_p\times[0,T[$} \\
Using once more the  notion of relative norm, we have by $(3.6)$:
\begin{equation}
||K_n(V_0,t)||\leq (K_{n,ij}K_n^{ij})^\frac{1}{2}||g_n(V_0,t)||.
\end{equation}
But $(K_{n,ij}K_n^{ij})_n$ and $(g_n)_n$ are uniformly bounded on
$\Omega_p\times[0,T[$. So is $(K_n)_n$.
\item[$3^0)$] \underline{Boundedness of $(\det g_n)_{n\in \mathbb{N}}$, $((\det g_n)^{-1})_n$ and $(g_n^{ij})_n$ on
$\Omega_p\times[0,T[$} \\
The relation $(6.37)$ is equivalent to:
\begin{equation*}
\det g=\det g^0\exp(-2\int_0^tH(V_0,s)ds).
\end{equation*}
By definition of the iterated sequence, we must have:
\begin{equation*}
\det g_{n+1}(V_0,t)=\det g^0\exp(-2\int_0^tH_n(V_0,s)ds),
\end{equation*}
From where, we deduce that, since $(H_n)_n$ is uniformly bounded on
the compact $\overline{\Omega}_p\times[0,T]$, there exist a constant
$C_{p,T}>0$ depending only on $p$ and $T$ such that:
\begin{equation*}
\det g^0\exp(-C_{p,T})\leq\det g_n(V_0,t)\leq\det
g^0\exp(C_{p,T}),\,\,\forall\,\,(V_0,t)\in \Omega_p\times[0,T[,
\end{equation*}
This shows that $(\det g_n)_n$ and $((\det g_n)^{-1})_n$ are
uniformly bounded on $\Omega_p\times[0,T[$. we can conclude that
$(g_n^{ij})_n$ is uniformly bounded on $\Omega_p\times[0,T[$.
\item[$4^0)$] \underline{Boundedness of $(E_n)_{n\in \mathbb{N}}$ on
$\Omega_p\times[0,T[$} \\
From $(8.3)$ we have:
\begin{equation*}
 16\pi(\tau_{n,00}+T_{n,00})=R_n+H_n^2-K_{n,ij}K_n^{ij}-2\Lambda\leq
H_n^2-2\Lambda.
\end{equation*}
Since $K_{n,ij}K_n^{ij}\geq0$, $R_n\leq0$.  But $(H_n)_n$ is
uniformly bounded on $\Omega_p\times[0,T[$. Then so are
$(\tau_{n,00})_n$ and $(T_{n,00})_n$. We know that:
\begin{equation*}
\tau_{n,00}=g_n^{ij}\tau_{n,ij}=\frac{1}{4}F_n^{ij}F_{n,ij}+\frac{1}{2}g_{n,ij}E_n^iE_n^j
\end{equation*}
then:
\begin{equation*}
0\leq\frac{1}{2}g_{n,ij}E_n^iE_n^j\leq\tau_{n,ij},
\end{equation*}
since $F_n^{ij}F_{n,ij}\geq0$. $(g_{n,ij}E_n^iE_n^j)_n$ is then
uniformly bounded on $\Omega_p\times[0,T[$. Applying $(3.9)$ and the
fact that $(g_{n,ij}E_n^iE_n^j)_n$ and $(g_n)_n$ are uniformly
bounded on $\Omega_p\times[0,T[$, we conclude that so is $(E_n)_n$.
\item[$5^0)$] \underline{Boundedness of $(F_n)_{n\in \mathbb{N}}$ on
$\Omega_p\times[0,T[$}\\
By the definition of the iterated sequence, equation $(2.53)$
implies:
\begin{equation}
F_{n,ij}(V_0,t)=F^0_{ij}+\int_0^tC^k_{ij}g_{n,kl}(V_0,s)E^l_n(V_0,s)ds.
\end{equation}
But $(g_n)_n$ and $(E_n)_n$ are uniformly bounded on
$\Omega_p\times[0,T[$; by $(8.9)$, $(F_n)_n$ is also uniformly
bounded on $\Omega_p\times[0,T[$.
\item[$6^0)$] \underline{Boundedness of $(u^0_n)_n$, $(\frac{1}{\rho_{_n}})_n$,
$(\rho_{_n})_n$ and $(u^i_n)_n$ on $\Omega_p\times[0,T[$}\\
We saw that $(2.28)$ and $(2.29)$ give:
\begin{equation*}
\frac{\dot{\rho}}{\rho}+
\frac{\dot{u}^0}{u^0}=-\left(\frac{3}{4}\frac{1}{u^0}-H+C^i_{ij}\frac{u^j}{u^0}
\right),
\end{equation*}
and by integration on $[0,t]$, $0<t<T$, we have:
\begin{equation*}
\rho u^0
=\rho_0u^0(0)\exp\left[-\int_0^t(\frac{3}{4}\frac{1}{u^0}-H+C^i_{ij}\frac{u^j}{u^0})ds\right]
\end{equation*}
By the definition of the iterated sequence, we must have:
\begin{equation}
\rho_{n+1}
u^0_{n+1}=\rho_0u^0(0)\exp\left[-\int_0^t(\frac{3}{4}\frac{1}{u^0_n}-H_n+C^i_{ij}\frac{u^j_n}{u^0_n})ds\right]
\end{equation}
By $(3.9)$ we have $\left| \frac{u^i}{u^0}\right|\leq C\left|
g_n\right|^\frac{3}{2}$, $(g_n)_n$ and  $(H_n)_n$ are uniformly
bounded on $\Omega_p\times[0,T[$, and $u^0_n\geq1$. So, from
$(8.10)$ we deduce that there exist a constant $C_{p,T}>0$ such
that:
\begin{equation*}
\rho_0u^{0}(0)\exp(-C_{p,T})\leq\rho_{n+1} u^0_{n+1}\leq
\rho_0u^{0}(0)\exp(C_{p,T}), \,\,\,\,on\,\,\Omega_p\times[0,T[.
\end{equation*}
So $(\rho_nu^0_n)_n$ and   $(\frac{1}{\rho_nu^0_n})_n$ are uniformly
bounded on $\Omega_p\times[0,T[$. But $u^0_n\geq1$, then
$(\rho_n)_n$ is also uniformly bounded on $\Omega_p\times[0,T[$. Now
we have by $(3.9)$:
\begin{equation*}
\left| \frac{u^i}{u^0}\right|=\left|
\frac{\rho_{n}u^i}{\rho_{n}u^0}\right|\leq C\left|
g_n\right|^\frac{3}{2}
\end{equation*}
But $(\rho_nu^0_n)_n$ and $(g_n)_n$ are uniformly bounded on
$\Omega_p\times[0,T[$.  Hence so is $(\rho_{n}u^i_n)_n$. Now
consider the equation in $u^i_{n+1}$ deduced from $(2.55)$. We
multiply this equation by $\rho_n$ and conclude that the sequence
$(\rho_{n}\dot{u}^i_{n+1})_n$ is uniformly bounded on
$\Omega_p\times[0,T[$. Now derive the relation
$u^0=\sqrt{1+g_{ij}u^iu^j}$ and obtain, using $(2.51)$:
\begin{equation*}
 \frac{\dot{u}^0}{u^0} =
-K_{ij}\frac{u^i}{u^0}
\frac{u^j}{u^0}+g_{ij}\frac{\rho\dot{u}^i}{\rho u^0} \frac{u^j}{u^0}
\end{equation*}
From there, we deduce that for the iterated sequence, we must have:
\begin{equation*}
 \frac{\dot{u}^0_{n+1}}{u^0_{n+1}} =
-K_{n,ij}\frac{u^i_n}{u^0_n}
\frac{u^j_n}{u^0_n}+g_{n,ij}\frac{(\rho_n\dot{u}^i_n)}{(\rho_n
u^0_n)} \frac{u^j_n}{u^0_n}
\end{equation*}
Since $(K_n)_n$, $(\frac{u^i_n}{u^0_n} )_n$, $(g_n)_n$, $(\rho_n
\dot{u}^i_n)_n$, $(\frac{1}{\rho_n u^0_n})_n$ are uniformly bounded
on $\Omega_p\times[0,T[$  there exists a constant $C_{p,T}>0 $ such
that:
\begin{equation*}
 \left|\frac{\dot{u}^0_{n+1}}{u^0_{n+1}}\right|\leq C_{p,T}
\end{equation*}
From there, we deduce that:
\begin{equation*}
1\leq u^0_{n+1}\leq U^{0,0}\exp(C_{p,T}T),
\,\,\,\,on\,\,\Omega_p\times[0,T[
\end{equation*}
 then $(u^0_{n+1})_n$  is uniformly bounded on
$\Omega_p\times[0,T[$. Now
$\frac{1}{\rho_n}=\frac{1}{\rho_nu^0_n}\times u^0_n$, thus
$(\frac{1}{\rho_n})_n$ is uniformly bounded on
$\Omega_p\times[0,T[$. We have $u^i_n=\frac{u^i_n}{u^0_n}\times
u^0_n$, and $\frac{u^i_n}{u^0_n}$, $ u^0_n$ are uniformly bounded on
$\Omega_p\times[0,T[$. Then $(u^i_n)_n$ is uniformly bounded on
$\Omega_p\times[0,T[$.\\
In conclusion, the sequences $(\rho_n)_n$, $(\frac{1}{\rho_n})_n$,
$(u^0_n)_n$ and $(u^i_n)_n$ are uniformly bounded on
$\Omega_p\times[0,T[$.
\item[$7^0)$] \underline{Boundedness of $(e_n)_{n\in \mathbb{N}}$ on
$\Omega_p\times[0,T[$} \\
The sequence $(e_n)_n$ from the  iterated sequence satisfies
constraint $(2.64)$, hence:
\begin{equation*}
e_n=-\frac{1}{u^0_n}C^i_{ik}E_n^{k}.
\end{equation*}
Since $u^0_n\geq1$ and $(E_n)_n$ is uniformly bounded on
$\Omega_p\times[0,T[$, so is $(e_n)_n$.
\item[$8^0)$] \underline{Boundedness of $(\theta_n^{0\alpha})_{n\in \mathbb{N}}$ on
$\Omega_p\times[0,T[$} \\
By the equation defining the iterated sequence, since $(H_n)_n$,
$(\rho_n)_n$, $(u^0_n)_n$, $(g_n^{ij})_n$, $(K_{n,ij})_n$,
$(\gamma^k_{n,ij})_n$ and $(u^i_n)_n$ are uniformly bounded on
$\Omega_p\times[0,T[$, we can conclude that, there exist two
constants $A_{p,T}>0$, and   $B_{p,T}>0$ such that:

\begin{equation*}
 \left|\theta_{n+1}^{00}(V_0,t)\right|\leq \left|\theta^{00}(V_0,0)\right|+
A_{p,T}\int_0^t\sum\limits_{\alpha=0}^3\left|\theta_n^{0\alpha}(V_0,s)\right|ds+B_{p,T}\,\,\,
on\,\,\Omega_p\times[0,T[
\end{equation*}
and
\begin{equation*}
\left|\theta_{n+1}^{0i}(V_0,t)\right|\leq
\left|\theta^{0i}(V_0,0)\right|+
A_{p,T}\int_0^t\sum\limits_{\alpha=0}^3\left|\theta_n^{0\alpha}(V_0,s)\right|ds+B_{p,T}\,\,\,
on\,\,\Omega_p\times[0,T[
\end{equation*}
From where we deduce that:
\begin{equation*}
\sum\limits_{\alpha=0}^3\left|\theta_{n+1}^{0\alpha}(V_0,t)\right|\leq
4A_{p,T}\int_0^t\sum\limits_{\alpha=0}^3\left|\theta_n^{0\alpha}(V_0,s)\right|ds+4B_{p,T}+
\sum\limits_{\alpha=0}^3\left|\theta^{0\alpha}(V_0,0)\right|\,\,\,
on \,\,\Omega_p\times[0,T[
\end{equation*}
By induction on $n$, we conclude that:
\begin{equation}
\sum\limits_{\alpha=0}^3\left|\theta_{n+1}^{0\alpha}(V_0,t)\right|\leq
\left(4A_{p,T}+\sum\limits_{\alpha=0}^3\left|\theta^{0\alpha}(V_0,0)\right|+4B_{p,T}\right)\exp(4TC_{p,T}),\,\,
on\,\,\Omega_p\times[0,T[.
\end{equation}
 Where $C_{p,T}$ is a constant. Now we know that
\begin{equation*}
\sum\limits_{\alpha=0}^3|\theta_0^{0\alpha}(V_0,0)|\leq|V_0|_1;
\end{equation*}
and $|V_0|_1$ is  continuous on the compact set
$\overline{\Omega}_p\times[0,T]$; then it is bounded on it. In
conclusion, $(8.11)$ shows that every sequence
$(\theta_{n}^{0\alpha})_n$ is uniformly bounded on
$\Omega_p\times[0,T[$. This ends the proof of proposition 8.2
\end{itemize}
\end{proof}
\begin{theorem}
The Cauchy problem for the Einstein-Maxwell system with
pseudo-tensor of pressure is well-posed in the sense of Hadamard.
\end{theorem}
\begin{proof}
We prove that the function:
$$\begin{array}{lcll}
S &:  \Omega&\longrightarrow &\mathcal{C}^1([0,\infty[,\widetilde{\Omega})\\
 & V_0&\longmapsto &      S(V_0,\,.\,\,),
\end{array}$$
is continuous. We know that it will be enough if we prove that its
restriction to every $\Omega_p$ is continuous.\\
Since $u^0_n\geq1$ and since the sequences $\left(V_n\right)_n$ and
$\left(\frac{1}{\rho_{_n}}\right)_n$ are uniformly bounded on
$\Omega_p\times[0,T[$, by the definition of the iterated sequence,
there exist a constant $C_{p,T}>0$ such that:
\begin{equation*}
\left|\dot{V}_{n+1}(V_0,t)-\dot{V}_{n+1}(\widetilde{V}_0,t)\right|_1\leq
C_{p,T}\left|V_n(V_0,t)-V_n(\widetilde{V}_0,t)\right|_1,\,\forall\,\,V_0,\,\widetilde{V}_0\in
\Omega_p,\,\,n\in\mathbb{N},\,t\in [0,T[.
\end{equation*}
We have:
\begin{equation*}
V_n(V_0,0)=V_0\,\,\, ; \,\,\,V_n(\widetilde{V}_0,0)=\widetilde{V}_0.
\end{equation*}
Then, integrating this inequation on $[0,t]$, $0<t<T$, we have:
\begin{equation*}
|V_{n+1}(V_0,t)-V_{n+1}(\widetilde{V}_0,t)|_1\leq
C_{p,T}\int_0^t\left|V_n(V_0,s)-V_n(\widetilde{V}_0,s)\right|_1ds+\left|V_0-
\widetilde{V}_0\right|_1
\end{equation*}
By induction on $n$, we deduce from this inequation that:
\begin{equation*}
|V_n(V_0,t)-V_n(\widetilde{V}_0,t)|_1\leq|V_0-\widetilde{V}_0|_1
\sum\limits_{i=0}^n\frac{(tC_{p,T})^i}{i!}
\end{equation*}
from where we deduce:
\begin{equation*}
|V_{n+1}(V_0,t)- V_{n+1}(\widetilde{V}_0,t)|_1\leq
|V_0-\widetilde{V}_0|_1\exp(TC_{p,T})
\end{equation*}
Taking the limit when $n\rightarrow+\infty$, we obtain:
\begin{equation*}
|S(V_0,t)-S(\widetilde{V}_0,t)|_1\leq
 |V_0-\widetilde{V}_0|_1\exp(TC_{p,T}),\forall\
\,V_0,\,\widetilde{V}_0\in\,\Omega_p \,\,,t\in[0,T],
\end{equation*}

Using the seminorm $P_{[0,T]}$ we obtain:
\begin{equation}
P_{[0;T]}(S(V_0,\,.\,)-S(\widetilde{V}_0,\,.\,))\leq
|V_0-\widetilde{V}_0|_1\exp(TC_{p,T}),\forall\,\,V_0,
\,\widetilde{V}_0\in\,\Omega_p.
\end{equation}
which proves that $V_0\longmapsto S(V_0,\,.\,)$ is continuous from
$\Omega_p$ to $\mathcal{C}^1([0,\infty[,\widetilde{\Omega})$.\\
This proves that the Cauchy problem for the Einstein-Maxwell system
with pseudo-tensor of pressure is well-posed.
\end{proof}

\newpage

\begin{center}
\underline{Conclusion and future investigation}
\end{center}
We have proved global existence of solutions to the evolution system
$(S)$ when the cosmological constant $\Lambda$ is strictly positive
and $H(0)<0$. We have also proved the geodesic completeness (when
$\sqrt{\frac{\Lambda}{3}}\geq\frac{3}{4}$), and determined the
asymptotic behavior (of the spacetimes) in the neighborhood of
$+\infty$ in the case of global existence. On the other hand, we
have finally proved that the problem is well posed in the sense of
Hadamard.

\newpage

\part*{APPENDIX}

\newpage

\begin{center}
\underline{Point 1 of Proposition 3.1}
\end{center}

\begin{itemize}
\item[$1)$]
\begin{center}
\underline{Evolution of
$A=R-K_{ij}K^{ij}+H^2-16\pi(\tau_{00}+T_{00})-2\Lambda$ }
\end{center}

\item[$1.1)$] \underline{Evolution of $H$}\\

From definition $(2.18)$ of $H$, we have:
\begin{equation*}
\partial_tH=\partial_t(g^{ij}K_{ij})=2K^{ij}K_{ij}+g^{ij}\partial_tK_{ij}
\end{equation*}
But equation $(2.52)$ gives:
\begin{eqnarray*}
g^{ij}\partial_tK_{ij}&=&g^{ij}\left[R_{ij}+HK_{ij}-2K^l_jK_{il}-
8\pi(\tau_{ij}+T_{ij})+4\pi(-T_{00}+g^{lm}T_{lm})g_{ij}-\Lambda g_{ij}\right] \\
&=& R+H^2-2K^{il}K_{il}-8\pi g^{ij}\tau_{ij} -8\pi
g^{ij}T_{ij}+12\pi(-T_{00}+g^{lm}T_{lm})-3\Lambda.
\end{eqnarray*}
From the two relations, we have the evolution of $H$:
\begin{equation*}
\partial_tH=R+H^2+4\pi g^{ij}T_{ij}-8\pi\tau_{00}-12\pi
T_{00}-3\Lambda.\,\,\,\,\,\,\,\,\,\,\,\,\,\,\,\,\,\,\,\,\,\,\,\,\,\,\,\,\,\,\,\,\,\,\,\,(A.1)
\end{equation*}

\item[$1.2)$]\underline{Evolution of $\tau_{00}+T_{00}$}\\

We use the index $(\,\,^4)$ to write the conservation laws:
\begin{eqnarray*}
^4\nabla_{\alpha}(\,\,^4\tau^{\alpha\beta}+\,^4T^{\alpha\beta})&=&\,\,^4g^{\alpha\lambda}
\,\,^4g^{\mu\beta}\,\,
^4\nabla_{\alpha}\,\,(\,\,^4\tau_{\lambda\mu}+\,^4T_{\lambda\mu})=0.
\end{eqnarray*}
Now, using (2.7) we have:
\begin{equation*}
\,\, ^{4}g^{\alpha\lambda}\,\, ^{4}g^{\mu\beta}\left[\,\,
^{4}e_{\alpha}(\,\, ^4\tau_{\lambda\mu}+ \,^4T_{\lambda\mu})-
\,\,^4\gamma_{\alpha\lambda}^{\theta}(\,\, ^4\tau_{\theta\mu}+
\,^4T_{\theta\mu})- \,\,^4\gamma_{\alpha\mu}^{\theta}(\,\,
^4\tau_{\lambda\theta}+ \,^4T_{\lambda\theta})\right]=0.
\end{equation*}
Then use (2.2), (2.3) and (2.19) to obtain the evolution of
$\tau_{00}+T_{00}$:
\begin{eqnarray*}
\frac{d}{dt}(\tau_{00}+T_{00})=
H(\tau_{00}+T_{00})-g^{ij}\gamma^k_{ij}(\tau_{0k}+T_{0k})+K^{ij}(\tau_{ij}+T_{ij})\,\,\,\,\,\,
\,\,\,\,\,\,\,\,\,\,\,\,\,\,\,\,\,\,\,\,\,\,\,\,(A.2).
\end{eqnarray*}

\item[$1.3)$] \underline{Evolution of $R$}\\
We have by definition:

$$\left\{
\begin{array}{lll}
R&=&g^{ij}R_{ij}\\
 ^{4}R^{\lambda}_{\,\,\alpha,\beta\delta} &=&\,\,
^{4}e_{\beta}(\,\,\,^{4}\gamma^{\lambda}_{\alpha\delta})-\,\,\,
^{4}e_{\delta}(\,\,\,^{4}\gamma^{\lambda}_{\beta\alpha}) +
\,\,\,^{4}\gamma^{\lambda}_{\beta\mu}\,\,\,^{4}\gamma^{\mu}_{\delta\alpha}
-\,\,\,^{4}\gamma^{\lambda}_{\delta\mu}\,\,\,^{4}\gamma^{\mu}_{\beta\alpha}-
\,\,^{4}C^{\mu}_{\beta\delta}\,\,\,^{4}\gamma^{\lambda}_{\mu\alpha}.\,\,\,\,\,\,\,\,\,(A.3)
\end{array}
\right.$$

 Where $C^{\mu}_{\beta\delta}$ are defined by:
\begin{equation*}
\left[\,\,^4e_{\alpha}\,,\,\,^4e_{\beta}\right]=\,\,^4C^{\lambda}_{\alpha\beta}\,\,^4e_{\lambda}
\,\,\,\,\,\,\,\,\,\,\,\,\,\,\,\,\,\,\,\,\,\,\,\,
\,\,\,\,\,\,\,\,\,\,\,\,\,\,\,\,\,\,\,\,\,\,\,\,\,\,\,\,\,\,\,\,\,\,
\,\,\,\,\,\,\,\,\,\,\,\,\,\,\,(A.4)
\end{equation*}

with:
\begin{equation*}
\,\,^4C^{0}_{\alpha\beta}\,\,= \,\,\,\,\,\,\,\,
^4C^{\lambda}_{0\beta} =\,\,^4C^{\lambda}_{\alpha0}=0;
\,\,\,\,^4C^{k}_{ij}=C^{k}_{ij}\,\,\,\,\,\,\,\,\,\,\,\,\,\,\,\,\,\,\,\,\,\,(A.5)
\end{equation*}

In particular, we deduce from (A.3):
\begin{eqnarray*}
^{4}R^{l}_{\,\,j,i0} &=& \,\,^{4}e_{i}(^{4}\gamma^{l}_{j0})-
\,\,^{4}e_{0}(^{4}\gamma^{l}_{ij}) +
\,\,^{4}\gamma^{l}_{i\tau}\,^{4}\gamma^{\tau}_{0j}
-\,\,^{4}\gamma^{l}_{0\tau}\,^{4}\gamma^{\tau}_{ij}-
\,\,^{4}C^{\tau}_{i0}\,^{4}\gamma^{l}_{\tau j}\\
 &=& -\frac{d}{dt}(\gamma^{l}_{ij})-
(K^k_j\gamma^l_{ik} -K^l_k\gamma^k_{ij}).
\end{eqnarray*}

Hence:
\begin{equation*}
^{4}R^{l}_{\,\,j,i0} = -\frac{d}{dt}(\gamma^{l}_{ij}) -\nabla_iK^l_j
\,\,\,\,\,\,\,\,\,\,
\,\,\,\,\,\,\,\,\,\,\,\,\,\,\,\,\,\,\,\,\,\,\,\,\,\,\,\,\,\,
\,\,\,\,\,\,\,\,\,\, \,\,\,\,\,\,\,\,\,\,\,\,\,\,\,(A.6)
\end{equation*}
Now we have from the Codazzi relation:
\begin{equation*}
^{4}R_{\lambda i,jl}\,n^{\lambda}=-\nabla_lK_{ij}+\nabla_jK_{il}.
\end{equation*}
But here we have $n=\partial_t$; this implies:
\begin{equation*}
^4R_{0i,jl}= -\nabla_lK_{ij}+\nabla_jK_{il}
\end{equation*}
or:
\begin{equation*}
^4R_{\,j,i0}^l=
-\nabla^lK_{ij}+\nabla_jK_{i}^l;\,\,\,\,\,\,\,\,\,\,\,\,\,\,\,\,
\,\,\,\,\,\,\,\,\,\,\,\,\,\,\,\,\,\,\,\,\,\,\,\,\,\,\,\,\,\,\,\,
\,\,\,\,\,\,\,\,\,\,\,\,\,\,\,\,\,\,\,\,\,\,\,\,\,\,\,\,\,\,\,\,(A.7)
\end{equation*}
(A.6) and (A.7) give:
\begin{eqnarray*}
\frac{d}{dt}(\gamma^{l}_{ij})&=&\nabla^lK_{ij}-\nabla_jK^l_i-\nabla_iK^l_j
\,\,\,\,\,\,\,\,\,\,\,\,\,\,\,\, \,\,\,\,
\,\,\,\,\,\,\,\,\,\,\,\,\,\,\,\,\,\,\,\, \,\,\,\,\,\,\,\,
\,\,\,\,\,\,\,\,\,\,\,\,\,\,\,\,(A.8).
\end{eqnarray*}
Now a direct calculation gives:
\begin{eqnarray*}
\nabla_l\left(\frac{d}{dt}\gamma^l_{ij}\right)&=&
\frac{d}{dt}R_{ij}.
\end{eqnarray*}

 We deduce from (A.8) that:
\begin{equation*}
\frac{d}{dt}R_{ij} = \nabla_l\nabla^l K_{ij} -\nabla_l\nabla_j
K_{i}^l-\nabla_l\nabla_i K_{j}^l
\end{equation*}
We deduce:
\begin{eqnarray*}
g^{ij}\frac{d}{dt}R_{ij}&=&-2g^{ln}\nabla_l\nabla^mK_{mn}\\
&=&- 2g^{ln}\left[\gamma^m_{li}\nabla^iK_{mn}
-\gamma^i_{lm}\nabla^mK_{in}- \gamma^i_{ln}\nabla^mK_{mi}\right]\\
&=&-2\left[\gamma^m_{li}\nabla^iK^l_m -\gamma^i_{lm}\nabla^mK^l_i
-g^{ln}\gamma^i_{ln}\nabla^mK_{mi} \right]
\end{eqnarray*}
and:
\begin{equation*}
g^{ij}\frac{d}{dt}R_{ij}=2g^{ij}\gamma^k_{ij}\nabla^lK_{lk}.
\end{equation*}
But:
\begin{equation*}
R=g^{ij}R_{ij} \Longrightarrow \partial_tR=2K^{ij}R_{ij} +
g^{ij}\frac{d}{dt}R_{ij}
\end{equation*}

 and the evolution of $R$:
\begin{equation*}
\partial_tR=2K^{ij}R_{ij}+
2g^{ij}\gamma^k_{ij}\nabla^lK_{lk}.\,\,\,\,\,\,\,\,\,\,\,\,\,\,\,\,\,\,\,
\,\,\,\,\,\,\,\,\,\,\,\,\,\,\,\,\,\,\,\,\,\,\,\,\,\,\,\,\,\,\,\,\,\,\,\,\,\,
\,\,\,\,\,\,\,\,\,\,\,\,\,\,\,\,\,\,\,(A.9)
\end{equation*}

\item[$1.4)$]\underline{Evolution of $K_{ij}K^{ij}$}\\
We have:
\begin{eqnarray*}
\partial_t(K_{ij}K^{ij})&=&K^{ij}\dot{K_{ij}}+K_{ij}\dot{K}^{ij}\\
&=&K^{ij}\dot{K_{ij}}
+K_{ij}\dot{\grave{\overbrace{\left[g^{ik}g^{jl}K_{kl}\right]}}}\\
 &=& K^{ij}\dot{K_{ij}}+  K_{ij}\left[4K^{ik}g^{jl}K_{kl}+g^{ik}g^{jl}\dot{K_{lk}}\right].
\end{eqnarray*}
Hence:
\begin{equation*}
\partial_t(K_{ij}K^{ij})= 2K^{ij}\dot{K_{ij}}+4K_{ij}K^{il}K_l^j.
\end{equation*}
Then equation (2.52) gives the evolution of $K_{ij}K^{ij}$:

\begin{equation*}
\partial_t(K_{ij}K^{ij})= 2K^{ij}R_{ij}+2HK_{ij}K^{ij} -16\pi
K^{ij}(\tau_{ij}+T_{ij}) + 8\pi H(-T_{00}+g^{lm}T_{lm})-2H\Lambda.
\,\,\,\,\,\,\,(A.10)
\end{equation*}
\item[$1.5)$]\underline{Evolution of $A$}\\
We have:
\begin{equation*}
\partial_tA= \partial_tR- \partial_t(K_{ij}K^{ij})+2H\partial_tH-16\pi
\partial_t(\tau_{00}+T_{00})
\end{equation*}

and using (A.9), (A.10), (A.1) and (A.2), we obtain the evolution of
$A$:
\begin{equation*}
\partial_tA=2HA+2g^{ij}\gamma^k_{ij}A_k.\,\,\,\,\,\,\,\,\,\,\,\,\,\,\,\,\,\,\,\,\,\,\,\,\,\,\,\,\,\,
\,\,\,\,\,\,\,\,\,\,\,\,\,\,\,\,\,\,\,\,
\,\,\,\,\,\,\,\,\,\,\,\,\,\,\,\,\,\,\,\,\,\,\,\,\,\,\,\,\,\,
\,\,\,\,\,\,\,\,(A.11)
\end{equation*}

\item[$2)$]
\begin{center}
\underline{Evolution of $A_j=\nabla^iK_{ij}+8\pi(\tau_{0j}+T_{0j})$}
\end{center}
We have to compute:
\begin{equation*}
 \frac{d}{dt}\left(\nabla^ik_{ij}\right),\,\,\,\,\,\,and\,\,\,\,\,\,\frac{d}{dt}\left(\tau_{0j}+T_{0j}\right).
\end{equation*}
If we take $\beta=j$ in the conservation laws (2.21), we obtain:
\begin{eqnarray*}
 ^4\nabla_\alpha(\,\,^4\tau^{\alpha j}+\,\, ^4T^{\alpha j})=\,\,^4\nabla_0
 (\,\,^4\tau^{0 j}+ \,\,^4T^{0 j})+\,\,^4\nabla_i(\,\,^4\tau^{i j}+\,\,^4T^{i
j})=0,
\end{eqnarray*}
or, using (2.7):
\begin{eqnarray*}
^4e_0(\,\,^4\tau^{0 j}+\,\, ^4T^{0 j}) &+&\,\, ^4\gamma^0_{0\lambda}
(\,\,^4\tau^{\lambda j}
+\,\, ^4T^{\lambda j}) \\
&+& \,\,^4\gamma^j_{0\lambda}(\,\,^4\tau^{0 \lambda} + \,\,^4T^{0
\lambda}) +\,\,^4e_i(\,\,^4\tau^{ij}+ \,\,^4T^{ij})+\,\,
^4\gamma^i_{i\lambda}(\,\,^4\tau^{\lambda j}+ \,\,^4T^{\lambda j})\\
 &+&\,\,
^4\gamma^j_{i\lambda} (\,\,^4\tau^{i\lambda }+ \,\,^4T^{i\lambda })\\
&=&0.
\end{eqnarray*}
From there we deduce, using (2.19) that:
\begin{equation*}
\frac{d}{dt}(\tau^{0j}+T^{0j})=H(\tau^{0j}+T^{0j})-\nabla^i(\tau^j_i+T^j_i)
+2K^j_i(\tau^{0i}+T^{0i}).\,\,\,\,\,\,\,\,\,\,\,\,\,\,\,\,\,\,\,\,\,\,(A.12)
\end{equation*}
Now:
\begin{eqnarray*}
\tau_{0n}+T_{0n}=
\,\,^4g_{0\alpha}\,\,\,^4g_{n\beta}(\,\,^4\tau^{\alpha\beta}+
\,\,^4T^{\alpha\beta}) =-g_{nj}(\tau^{0j}+T^{0j})
\end{eqnarray*}
Then, using $\frac{d}{dt}g_{nj}=-2K_{nj}$ given by equation (2.51),
we obtain:
\begin{eqnarray*}
\frac{d}{dt}(\tau_{0n}+T_{0n}) =2K_{nj}(\tau^{0j}+T^{0j})
-g_{nj}\frac{d}{dt}(\tau^{0j}+T^{0j}).
\end{eqnarray*}
We then use (A.12) to obtain:
\begin{equation*}
\frac{d}{dt}(\tau_{0j}+T_{0j})=H(\tau_{0j}+T_{0j})+\nabla^i(\tau_{ij}+T_{ij}).
\,\,\,\,\,\,\,\,\,\,\,\,\,\,\,\,\,\,\,\,\,\,\,\,\,\,\,\,\,\,\,\,(A.13)
\end{equation*}
Now we have by Bianchi identities:
\begin{eqnarray*}
   ^4\nabla_\alpha\,\,^4R^{\alpha\beta}-\frac{1}{2}\,\,^4\nabla^\beta
   \,\,^4R=0.
\end{eqnarray*}
But  using  (2.19), with $\beta=j$, we have:
\begin{eqnarray*}
   ^4\nabla^\beta\,\,^4R=\,\,^4\nabla^\beta\,\,^4R_\alpha^\alpha=\,\,
 g^{ij}\,\,^4\nabla_i\,\,^4R_i^\alpha  =\,\,
   \,^4g^{ij}\left[\,^4\gamma^\alpha_{i\lambda}\,^4R^\lambda_\alpha-
    \,^4\gamma^\lambda_{i\alpha}\,^4R^\alpha_\lambda\right]=0.
\end{eqnarray*}
Then we have:
\begin{equation*}
^4\nabla_\alpha\,^4R^{\alpha j}=0.
\end{equation*}
But:
\begin{eqnarray*}
(^4\nabla_\alpha\,^4R^{\alpha j}=0)
\Longrightarrow(\,^4\nabla_0\,^4R^{0j}+\,^4\nabla_i\,^4R^{ij}=0),
\end{eqnarray*}
and developping, we obtain:
\begin{eqnarray*}
\frac{d}{dt}(\,^4R^{0j}) -\,K^j_{i}\,^4R^{0i} +
\,^4\nabla_i\,^4R^{ij}=0.
\,\,\,\,\,\,\,\,\,\,\,\,\,\,\,\,\,\,\,\,\,\,\,\,\,\,\,\,\,\,\,\,
\,\,\,\,\,\,\,\,\,\,\,\,\,\,\,\,\,\,\,\,\,\,\,\,\,\,\,\,\,\,\,\,(A.14)
\end{eqnarray*}
This means that we have to compute
\begin{eqnarray*}
\,^4R^{0i} \,\, and\,\,^4\nabla_i\,^4R^{ij}.
\end{eqnarray*}
We know that:
\begin{eqnarray*}
 ^4R
_{\alpha\beta}&=&\,^4e_\lambda\left(\,^4\gamma^\lambda_{\alpha\beta}\right)-\,\,
^4e_\beta\left(\,^4\gamma^\lambda_{\lambda\alpha}\right)
+\,\,^4\gamma^\lambda_{\lambda\mu}\,^4\gamma^\mu_{\,\beta\alpha}
-\,^4\gamma^\lambda_{\beta\mu}\,^4\gamma^\mu_{\,\lambda\alpha}
-C^\mu_{\lambda\beta}\,^4\gamma^\lambda_{\mu\alpha}.\,\,\,\,\,\,\,\,\,\,\,\,
\,\,\,\,\,\,\,\,\,\,\,\,\,\,\,\,\,\,\,\,\,\,\,\,(A.15)
\end{eqnarray*}
We then obtain, using once more (2.19):

\begin{eqnarray*}
^4R^{0j}&=&\,\,^4g^{0\alpha}\,\,^4g^{j\beta}\,\,^4R_{\alpha\beta}\\
&=&-g^{ij}\,\,^4R_{0i}\\
&=&g^{ij}\left[K^{k}_i\,^4\gamma^l_{lk}-
K^{l}_k\,\,^4\gamma^k_{li}\right]\\
&=&g^{ij}\nabla_lK_i^l.
\end{eqnarray*}
Thus:
\begin{equation*}
^4R^{0j}=\nabla_iK^{ij}.\,\,\,\,\,\,\,\,\,\,\,\,\,\,\,\,\,\,\,\,\,\,
\,\,\,\,\,\,\,\,\,\,\,\,\,\,\,\,\,\,\,\,\,\,(A.16)
\end{equation*}
If we consider formula (A.15), we deduce using (2.19) that:
\begin{eqnarray*}
^4R^{ij}&=&\,^4g^{ik}\,^4g^{j\,l}\,^4R_{kl}\\
&=&-g^{ik}g^{j\,l}\partial_tK_{kl}+ HK^{ij}+ g^{ik}g^{jl}\left[
\,^4\gamma^p_{pq}\,^4\gamma^q_{lk}
-\,^4\gamma^p_{lq}\,^4\gamma^q_{pk}
-C^p_{ql}\,^4\gamma^q_{pk}\right] - 2K_p^jK^{pi}.
\end{eqnarray*}
i.e.,   using (2.60):
\begin{equation*}
^4R^{ij}=R^{ij}+HK^{ij}- 2K_p^jK^{pi}-g^{ik}g^{jl}\partial_tK_{kl}.
\,\,\,\,\,\,\,\,\,\,\,\,\,\,\,\,\,\,\,\,\,\,\,\,\,\,\,\,\,\,(A.17)
\end{equation*}
We then deduce, using (A.16):
\begin{eqnarray*}
^4\nabla_i\,\,^4R^{ij}=\,\,^4\gamma^i_{i\lambda}\,\,^4R^{\lambda j}
+\,\,^4\gamma^j_{i\lambda}\,\,^4R^{i\lambda}=-H\,\,^4R^{0j}-
\,\,K^j_{i}\,\,^4R^{0i}+ \,\,^4\gamma^i_{ik}\,\,^4R^{kj}+
\,\,^4\gamma^j_{ik}\,\,^4R^{ik}
\end{eqnarray*}
and we deduce from (A.17) and (2.19):
\begin{equation*}
^4\nabla_i\,^4R^{ij}=- K^j_i\nabla_kK^{ki}+ \nabla_iR^{ij}
-2\nabla_i(K^i_pK^{pj})- g^{ip}g^{jq}\nabla_i(\partial_tK_{pq}).
\,\,\,\,\,\,\,\,\,\,\,\,\,\,\,\,\,\,\,\,\,\,\,\,(A.18)
\end{equation*}
We obtain from (A.14), (A.16) and (A.18):
\begin{eqnarray*}
\frac{d}{dt}(\nabla_iK^{ij})-K^j_i\nabla_kK^{ik} -
K^j_i\nabla_kK^{ik} + \nabla_iR^{ij} -2\nabla_i(K^i_pK^{pj}) -
g^{ip}g^{jq}\nabla_i(\partial_tK_{pq})=0
\end{eqnarray*}
We now use equation (2.52) to express $\partial_tK_{pq}$ and we
obtain:
\begin{eqnarray*}
\frac{d}{dt}(g^{jk}\nabla^iK_{ik})- 2K^j_i\nabla_kK^{ik}+
\nabla_iR^{ij} -2\nabla_i(K^i_pK^{pj})-\\ g^{ip}g^{jq}\nabla_i
\left[R_{pq}+ HK_{pq} -2K_q^kK_{pk} -8\pi(\tau_{pq}+T_{pq})
+4\pi(-T_{00}+g^{lm}T_{lm})g_{pq} -\Lambda g_{pq}\right]=0
\end{eqnarray*}
Or, since the functions depend only on $t$, and $
\frac{d}{dt}g^{jk}=2K^{jk}$:
\begin{eqnarray*}
g^{j\,k}\frac{d}{dt}(\nabla^iK_{ik}) +2K^{j\,k}\nabla^iK_{ik}
-2K^j_i\nabla_kK^{ik}&+&\nabla_iR^{ij} -2\nabla_i(K^i_pK^{pj})\\
-\nabla_iR^{ij}-H\nabla_iK^{ij}+2\nabla_i(K^{kj}K^i_k)
&+&8\pi\nabla_i(\tau^{ij}+T^{ij})=0
\end{eqnarray*}
or:
\begin{eqnarray*}
g_{jn}\left[g^{jk}\frac{d}{dt}(\nabla^iK_{ik})-H\nabla_iK^{ij}
+8\pi\nabla_i(\tau^{ij}+T^{ij})\right]=0,
\end{eqnarray*}
that is:
\begin{equation*}
\frac{d}{dt}(\nabla^iK_{ij})=H\nabla^iK_{ij}-
8\pi\nabla^i(\tau_{ij}+T_{ij}).\,\,\,\,\,\,\,\,\,\,\,\,\,\,\,\,\,\,\,\,\,\,
\,\,\,\,\,\,\,\,\,\,\,\,\,\,\,\,\,\,\,\,\,\,(A.19)
\end{equation*}

We finally obtain, using (A.13) and (A.19):
\begin{eqnarray*}
\frac{d}{dt}A_j&=& \frac{d}{dt}(\nabla^iK_{ij})+8\pi\frac{d}{dt}
(\tau_{0j}+T_{0j})\\
&=& H\nabla^iK_{ij}- 8\pi\nabla^i(\tau_{ij}+T_{ij})
+8\pi\left[H(\tau_{0j}+T_{0j})+\nabla^i(\tau_{ij}+T_{ij})\right]\\
&=&H\left[\nabla^iK_{ij}+8\pi(\tau_{0j}+T_{0j})\right].
\end{eqnarray*}
and the evolution of $A_j$:
\begin{equation*}
\frac{d}{dt}A_j=HA_j
\end{equation*}
\item[$3)$]
\begin{center}
\underline{Evolution of $A_{ijk}=C^l_{ij}F_{kl}+C^l_{jk}F_{il}+
C^l_{ki}F_{jl}$}
\end{center}

Using equation (2.53) in $F_{ij}$, we have:
\begin{eqnarray*}
\partial_t(A_{ij\,k})&=&C^l_{ij}\partial_t(F_{kl})+C^l_{j\,k}\partial_t(F_{il})+
C^l_{ki}\partial_t(F_{j\,l})\\
&=&C^l_{ij}\left(C^m_{kl}g_{ml}E^l\right)+
C^l_{jk}\left(C^m_{il}g_{ml}E^l\right)+
C^l_{ki}\left(C^m_{j\,k}g_{ml}E^l\right)\\
&=&\left(C^l_{ij}C^m_{kl}+C^l_{j\,k}C^m_{il}+
C^l_{ki}C^m_{j\,k}\right)g_{ml}E^l\\
 &=&0,
\end{eqnarray*}
from the equality of Jacobi. Hence:
\begin{equation*}
\partial_tA_{ij\,k}=0.
\end{equation*}

\item[$4)$]
\begin{center}
 \underline{Evolution of $B=C^i_{ik}F^{0k}+ eu^0$}
\end{center}
Set
\begin{eqnarray*}
B^\alpha = C^i_{ik}F^{\alpha k}+ eu^\alpha.
\end{eqnarray*}
We have $B^0=B$. Setting $J^k=eu^k$, we have from equation (2.9):
\begin{eqnarray*}
^4\nabla_\alpha B^\alpha&=&
C^i_{ik}\,^4\nabla_\alpha\,\,^4F^{\alpha k}\\
 &=&C^i_{ik}J^k.
\end{eqnarray*}
Thus:
\begin{eqnarray*}
^4\nabla_0B^0+^4\nabla_iB^i=C^i_{ik}J^k.
\end{eqnarray*}
or:
\begin{eqnarray*}
\frac{d}{dt}B+\,\,^4\gamma^i_{i\lambda}B^\lambda
 &=&C^i_{ik}J^k;
\end{eqnarray*}
i.e.
\begin{eqnarray*}
\frac{d}{dt}B+\,\,^4\gamma^i_{i0}B+ \,^4\gamma^i_{i\,j}B^j
=C^i_{i\,k}J^k;
\end{eqnarray*}

Use (2.20) to obtain:
\begin{eqnarray*}
\frac{d}{dt}B -HB
+C^i_{ij}\left[C^l_{lk}F^{j\,k}+eu^j\right]=C^i_{i\,k}eu^k.
\end{eqnarray*}
or:
\begin{eqnarray*}
\frac{d}{dt}B -HB+
C^i_{ij}C^l_{lk}F^{j\,k}+C^i_{ij}eu^j=C^i_{ik}eu^k.
\end{eqnarray*}
But
\begin{eqnarray*}
C^i_{ij}C^l_{lk}F^{j\,k}=0.
\end{eqnarray*}
Hence:
\begin{equation*}
\frac{d}{dt}B-HB=0
\end{equation*}
and the evolution of B:
\begin{equation*}
\frac{d}{dt}B=HB.
\end{equation*}
This ends the appendix.\,\,\,\,\,\,\,\,$\Box$
\end{itemize}

\end{document}